\documentclass[sigconf,nonacm]{acmart}
\usepackage{enumitem}
\usepackage{booktabs,multirow,array,tabularx}
\usepackage{xcolor,colortbl}
\usepackage{tcolorbox}

\newtcolorbox{findingsbox}[1]{
    title={#1},
    fonttitle=\bfseries,
    fontupper=\small,
    colback=black!2,
    colbacktitle=black!7,
    coltitle=black,
    colframe=black!30,
    boxrule=0.4pt,
    arc=1pt,
    left=5pt,
    right=5pt,
    top=4pt,
    bottom=4pt,
    before skip=6pt,
    after skip=6pt
}
\AtBeginDocument{%
  }

\setcopyright{acmlicensed}
\copyrightyear{2027}
\acmYear{2027}
\acmISBN{978-1-4503-XXXX-X/2018/06}

\begin{document}

\title{The Vocabulary Is Already Here: Adapting Pretrained Interpretable Features for Generative Recommendation}

\author{Ziluowen Luo}
\affiliation{%
  \institution{Central South University}
  \city{Changsha}
  \state{Hunan}
  \country{China}
}
\email{lzlwddl@csu.edu.cn}

\author{Ruochen Liu}
\affiliation{%
  \institution{The Hong Kong University of Science and Technology (Guangzhou)}
  \city{Guangzhou}
  \country{China}
}
\email{rliu190@connect.hkust-gz.edu.cn}

\author{Jun Yin}
\affiliation{%
  \institution{Hong Kong Polytechnic University}
  \state{Hongkong}
  \country{China}
}
\email{Junmay.yin@connect.polyu.hk}

\author{Hao Yan}
\affiliation{%
  \institution{Central South University}
  \city{ChangSha}
  \state{Hunan}
  \country{China}
}
\email{csuyh1999@csu.edu.cn}

\author{Zhengxuan Hu}
\affiliation{%
  \institution{Queen Mary University of London}
  \state{London}
  \country{UK}
}
\email{jp2024213778@qmul.ac.uk}

\author{Xinze Shi}
\affiliation{%
  \institution{ Queen Mary University of London}
  \state{London}
  \country{UK}
}
\email{jp2024213114@qmul.ac.uk}

\author{Chaozhuo Li}
\affiliation{%
  \institution{Microsoft}
  \country{China}
}
\email{lichaozhuo1991@gmail.com}



\author{Chengqi Zhang}
\affiliation{%
  \institution{Hong Kong Polytechnic University}
  \state{Hongkong}
  \country{China}
}
\email{chengqi.zhang@polyu.edu.hk}

\author{Senzhang Wang}
\correspondingauthor
\affiliation{%
  \institution{Central South University}
  \state{Changsha}
  \country{China}
}
\email{szwang@csu.edu.cn}

\renewcommand{\shortauthors}{Trovato et al.}

\begin{abstract}
Generative recommendation (GR) follows the token generation paradigm of language models, predicting the next item through its semantic ID (SID). 
Existing works commonly fit SID vocabularies to pretrained language model (PLM) representations of item descriptions, incurring additional training costs and making vocabulary coverage dependent on the fitting domains. 
Alternatively, we observe that pretrained sparse autoencoders (P-SAEs) expose PLM features associated with human-interpretable semantic annotations, allowing feature indices to serve as a natural SID vocabulary. 
However, reusing this vocabulary poses two challenges. 
First, \textit{recommendation adaptation} requires preserving item distinctions under compact token budgets and learning to favor the target item over candidates with shared dominant features. 
Second, \textit{retrieval explanation} requires extending feature interpretability from item content to retrieval decisions. 
To address these challenges, we propose \textbf{\textit{API}}, a framework that \underline{A}dapts the \underline{P}retrained \underline{I}nterpretable vocabulary for effective and explainable GR without additional vocabulary learning. 
Within API, Pretrained Interpretable IDs (PIDs) provide compact item representations that retain activateed feature indices and relative activation strengths for item discrimination. 
With these PIDs as supervision, the Semantic Prediction Adapter (SPA) learns user preferences from interaction histories through joint feature-index and strength prediction.
SPA then supervises agreement with the target PID and encourages the observed next item to outrank its alternatives. 
Using this score for retrieval, API further explains candidate score differences through feature contributions paired with their pretrained semantic annotations. 
Extensive experiments show that API achieves competitive performance across backbones and decoding schemes, while preserving item distinctions under the compact token budget and enabling explainable retrieval.
\end{abstract}



\received{20 February 2007}
\received[revised]{12 March 2009}
\received[accepted]{5 June 2009}

\maketitle

\noindent \textbf{Relevance:} 
This work is relevant to the Track of \textit{User modeling, personalization and recommendation}. 
This work contributes to recommendation via adapting the pretrained vocabulary and provides a practical path for applying interpretability to deployment. 

\section{Introduction}\label{sec:intro}
\begin{figure}[tbp]
  \centering
  \includegraphics[width=\columnwidth]{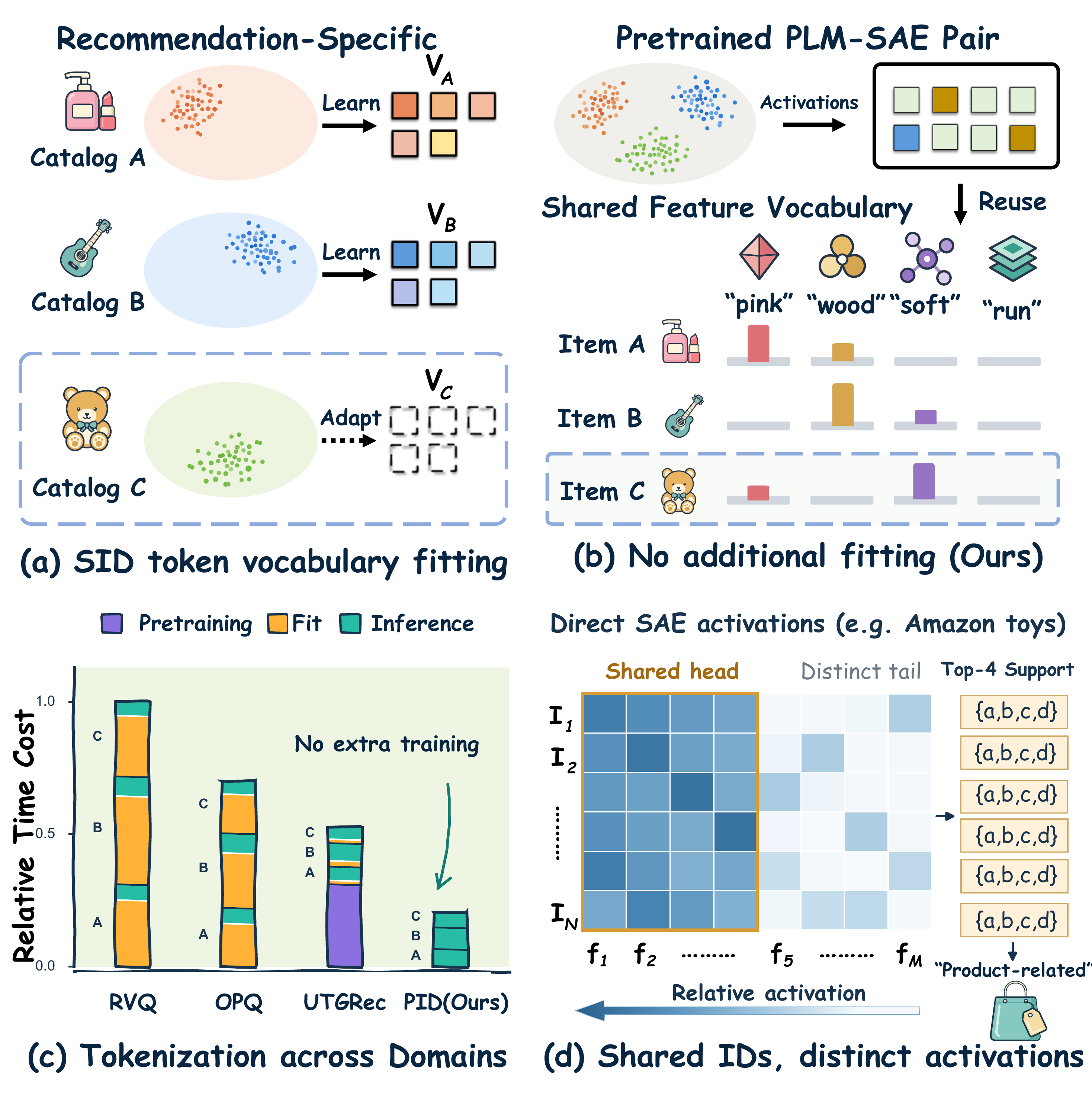}
\caption{The vocabulary is already here. 
(a--c) Pretrained SAEs enable vocabulary sharing without extra training.
(d) Distinct items can share the same dominant feature indices.}
  \label{fig:fig1}
  \vspace{-0.3cm}
\end{figure}

Generative recommendation systems (GRs) adopts the token generation paradigm of language models (LMs) to retrieve items conditioned on the user's interaction history~\cite{liu2026setllm,liu2024fine,luo2026otk}. 
A widely used approach represents items as the semantic IDs (SIDs), each consisting of a short sequence of discrete tokens. 
The mapping from items to SIDs, termed \textit{item tokenization}, affects the GR pipeline from two fundamental perspectives: \textit{item representation} and \textit{prediction supervision}~\cite{zhao2026simgr}. 
The resulting SIDs determine which item distinctions are retained and define the token-level predictive targets used to train the recommendation model. 
Therefore, a central problem in tokenization is how to construct compact SIDs that preserve item distinctions and provide effective supervision for recommendation. 

Existing approaches commonly construct SIDs by fitting a special discrete vocabulary over item representations produced by pretrained language models (PLMs)~\cite{rajput2023recommender,hou2025generating,zhu2024cost}. 
Inspired by the strong generalizability of unified natural language vocabularies in LMs , recent efforts have explored shared vocabularies for generalizable item tokenization~\cite{geng2022recommendation,xue2021mt5}. 
For example, UTGRec~\cite{zheng2026universal} learns a transferable tokenizer through multi-domain pretraining, while UniTok~\cite{hou2026tokenize} combines shared and specific experts for unified tokenization. 
Though these advances broaden tokenizer generalizability, they still construct vocabulary through additional fitting on the recommendation-specific data. 
The semantic coverage of the learned vocabulary therefore depends on the fitting data, while vocabulary construction incurs additional optimization costs.
Moreover, fitting a vocabulary to recommendation-specific data does not necessarily yield SID tokens with readily interpretable semantics~\cite{fang2025hidvae}.
These observations raise a natural question: \textit{Is there a free lunch for vocabulary construction in generative recommendation?}

On the basis of recent advances in \textit{mechanistic interpretability}, we answer this question in the affirmative. 
Specifically, pretrained sparse autoencoders (P-SAEs) already provide feature dictionaries for their paired PLMs, with many features admitting human-interpretable descriptions~\cite{gao2025scaling}. 
For each item, the P-SAE encodes its PLM representation into a sparse activation vector, whose nonzero coordinates identify the expressed features and whose values indicate their strengths. 
For example, a red T-shirt might activate a feature associated with "red" more strongly than one associated with "blue". 
As illustrated in Figure~\ref{fig:fig1}(b--c), the indices of these features can serve as SID tokens, providing a natural vocabulary that can be reused without additional recommendation-specific vocabulary learning. 
This vocabulary also carries inspectable semantics that P-SAE features often respond to recognizable semantic patterns in their inputs, allowing the corresponding tokens to indicate which properties are expressed in an item. 

However, bringing this vocabulary int effective and explainable GR still requires addressing two key challenges. 
First, \textit{\textbf{recommendation adaptation}} requires preserving item distinctions in both SID construction and the supervision built on these SIDs. 
As illustrated in Figure~\ref{fig:fig1}(d), similar items can share dominant features while differing in their activation strengths and less active features. 
Representing items only by their most active feature indices can discard these distinctions, while predicting shared indices alone does not ensure that the target item outranks similar candidates. 
Moreover, while these features make item semantics inspectable, unleashing the potential of the vocabulary requires extending this interpretability to the retrieval. 
This motivates the second challenge, \textit{\textbf{retrieval explanation}}, which is to trace feature contributions to items and explain why the model favors one item over another.

To address these challenges, we propose \textbf{\textit{API}}, a generative recommendation framework that \underline{A}dapting the \underline{P}retrained \underline{I}nterpretable vocabulary for effective and explainable generative recommendation. 
This adaptation spans two complementary levels: \textit{item representation} through \underline{P}retrained interpretable \underline{ID} (\textbf{PID}) construction, and \textit{recommendation learning} through a \underline{S}emantic \underline{P}rediction \underline{A}dapter (\textbf{SPA}). 
PID construction prioritizes strongly activated features while discounting those prevalent across the catalog, producing compact identifiers with accompanying activation evidence. 
Building on these representations, SPA incorporates feature identities and relative activation strengths into history modeling and aligns feature prediction of the target. 
To supervise these predictions, SPA combines sequence likelihood and strength agreement into retrieval and introduces preference alignment against competing items. 
API further traces feature-level contributions to retrieval, linking the semantics of each feature to the predicted evidence that favors a retrieved item. 
Our contributions are summarized as:
\vspace{-0.1cm}
\begin{itemize}
    \item We propose API, a generative recommendation framework that replaces recommendation-specific vocabulary learning with pretrained concept adaptation, opening a practical path for applying interpretability to practical deployment.

    \item We introduce PIDs that preserve item distinctions through feature selection and relative strengths, and theoretically characterize their separability under compact token budget.

    \item We develop SPA to effective and interpretable recommendation, aligning user preferences through item-level supervision and explain rankings through feature contributions.

    \item Experiments substantiate pretrained vocabulary reuse as a viable basis for GR, showing that API support competitive recommendation and enables explainable retrieval. Code is available at \url{https://anonymous.4open.science/r/API-C1F9/}.
\end{itemize}
\section{Related Works}\label{sec:related_work}
\noindent\textbf{Generative recommendation and item tokenization.}
Generative recommendation (GR) makes item identifiers part of the prediction space, placing item tokenization at the interface between item representation and model supervision. 
Language-based approaches include the text-to-text formulation~\cite{geng2022recommendation} and the learned textual identifiers~\cite{tan2024idgenrec}. 
For structured identifiers, TIGER constructs semantic IDs (SIDs) through residual quantization~\cite{lee2022autoregressive} of pretrained language model (PLM) representations~\cite{rajput2023recommender} while SEATER combines balanced tree indexing with contrastive learning~\cite{si2024generative}. 
LC-Rec then incorporate collaborative semantics into item indexing~\cite{zheng2024adapting,wang2024learnable} and ETEGRec couples tokenizer learning with recommendation through alternating optimization~\cite{liu2025generative}. 
UTGRec and UniTok further extend vocabulary learning across domains through shared tokenization~\cite{zheng2026universal,hou2026tokenize}. 
Despite these advances, the vocabulary learning still relies on recommendation-specific data which requires additional optimization and its semantic organization remains limited by the training domains. 
Our work extends pretrained knowledge reuse from continuous representations to the vocabulary units, adapting an existing PLM feature dictionary for recommendation.

\textbf{Sparse autoencoders and pretrained feature reuse.}
Motivated by the \textit{mechanistic interpretability}~\cite{bricken2023towards,huben2024sparse}, sparse autoencoders (SAEs) decompose language model activations into sparse combinations of dictionary features, exposing interpretable structure within dense representations~\cite{luo2026otk}. 
Subsequent work scales these dictionaries~\cite{templeton2024scaling,gao2025scaling}, improves reconstruction under sparsity constraints~\cite{rajamanoharan2024jumping} and releases pretrained SAEs for reuse through resources such as Gemma Scope~\cite{lieberum2024gemma}. 
In recommendation, SAERec uses SAE-derived intent priors to enrich sequence modeling~\cite{xia2026saerec}, but does not explore feature-level explanations of generative retrieval. 
In this work, we explores a practical route from mechanistic interpretability to recommendation systems, reusing a pretrained SAE feature dictionary as a shared vocabulary for GR. 
This further requires adapting reconstruction-oriented features for recommendation supervision and connecting feature-level readable concepts to item retrieval.
\section{Preliminaries}\label{sec:preliminary}

\textbf{Generative recommendation.}
Let $\mathcal I$ denote an item set, where each item $i\in\mathcal I$ is associated with a textual description $x_i$~\cite{luo2026distilling,luo2026nice}. 
Given a user’s sequential interaction history $\mathcal H_u=(i_1,\ldots,i_T)$, the task is to predict the subsequent item $i^+\in\mathcal I$. 
Generative recommendation (GR) approaches this task by generating a discrete token sequence conditioned on $\mathcal H_u$ and mapping the generating sequence to items. 
Autoregressive (AR) decoding generates tokens sequentially, conditioned on preceding tokens~\cite{rajput2023recommender}, whereas parallel (PR) decoding predicts multiple token in a single decoding step~\cite{hou2025generating}.

\textbf{Item tokenization.} 
For text-based item tokenization~\cite{geng2022recommendation}, a pretrained language model (PLM) (e.g., GPT-2~\cite{radford2019language}, T5~\cite{raffel2020exploring}) encodes the item description $x_i$ into continuous hidden representations $h_i$. 
The tokenizer then converts them into a semantic ID (SID), a compact sequence of discrete tokens drawn from a shared set $\mathcal V$, referred to as the \textit{vocabulary}. 
Formally, the tokenizer $\tau$ is defined:
\begin{equation}
\tau:\mathcal I\rightarrow\mathcal V^K,
\qquad
c_i=\tau(i)=(c_{i,1},\ldots,c_{i,K}),
\end{equation}
where $K$ is the length and $c_{i,r}\in\mathcal V$ denotes the token at position $r$. 
The vocabulary thus determine the information retained in these tokens and the token-level supervision used for recommendation.

\textbf{Pretrained sparse autoencoders.} 
A pretrained sparse autoencoders (P-SAE) decomposes continuous PLM representations into sparsely activated dictionary features \cite{gao2025scaling}. 
Given a hidden representation $\mathbf h$ from its paired PLM, the encoder $E$ produces a non-negative activation vector $\mathbf z=E(\mathbf h)$. 
The corresponding reconstruction is
\begin{equation}
    \hat{\mathbf h}
    =
    \mathbf D\mathbf z+\mathbf b_{\mathrm{dec}},
\end{equation}
where $\mathbf D=[\mathbf d_1,\ldots,\mathbf d_m]$ is the decoder dictionary and $\mathbf b_{\mathrm{dec}}$ is the decoder bias. 
Each column $\mathbf d_j$ represents a feature identified by its index $j$ and the activation $z_j$ indicates how strongly that feature is expressed in $\mathbf h$. 
The associated pretrained semantic annotation $m_j$ provides a human-readable description of the feature (e.g., "red"). 
These feature indices can thus serve as a reusable token vocabulary, accompanied by activation strengths and semantic annotations.

\section{Methodology}\label{sec:method}
\begin{figure*}[t]
    \centering
    \includegraphics[width=\linewidth]{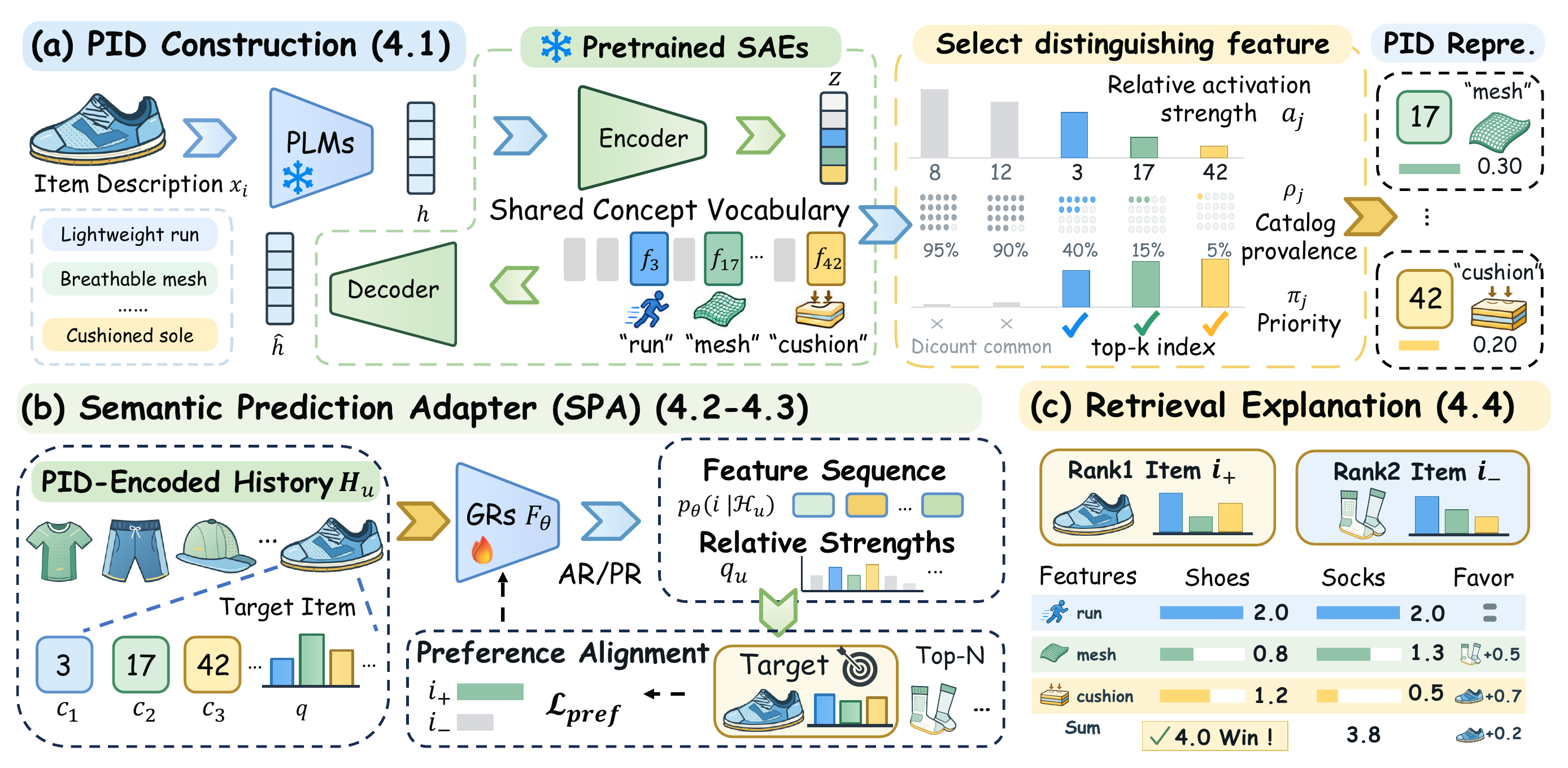}
\caption{Overview of API. 
(a) PID construction selects pretrained features and retains their
relative activation strengths and semantic annotations.
(b) SPA learns user preferences through feature-index and strength
prediction with item-level supervision.
(c) Candidate score differences are decomposed into semantic
feature contributions for retrieval explanation.}
    \label{fig:api-pipeline}
    \vspace{-0.3cm}
\end{figure*}

In this section, we propose API, a generative recommendation framework comprising Pretrained Interpretable ID (PID) construction and the Semantic Prediction Adapter (SPA). 
As illustrated in Figure~\ref{fig:api-pipeline}, PID construction selects pretrained features and retains their relative activation strengths to mitigate collisions(\S~\ref{sec:4.1}). 
Building on this representation, SPA then learns user preferences from interaction histories through PID prediction (\S~\ref{sec:spa-prediction}), with item-level supervision encouraging the GR model to favor the target item over other candidates (\S~\ref{sec:spa-competition}). 
We further extend explanation from item representations to retrieval decisions, revealing the  evidence behind the preference for one candidate over another (\S~\ref{sec:retrieval-explanation}).

\subsection{Pretrained Interpretable IDs Construction}
\label{sec:4.1}
To construct PIDs that preserve item distinctions, we first identify
the pretrained features expressed in each item and quantify their
strengths. 
For an item description sentence $x_i$, the frozen PLM--SAE pair\footnote{
We use the pretrained SAE released at \url{https://decoderesearch.github.io/SAELens/latest/pretrained_saes/}, together with the PLM and layer used for SAE pretraining.}
produces sparse activations $\mathbf{z}_i^{(t)}$ over token $t$. 
We then aggregate the activations into an item-level strength for
feature $j$: 
\begin{equation}
    a_{ij}
    =
    \frac{\|\mathbf{d}_j\|_2}{|\mathcal{T}_i|}
    \sum_{t\in\mathcal{T}_i} z_{i,j}^{(t)},
    \label{eq:pid-activation}
\end{equation}
where $\mathcal{T}_i$ denotes the token positions in $x_i$, excluding special token. 
The resulting $a_{ij}$ quantifies the item-level activation strength of feature $j$, indicating how strongly the item expresses that feature.

However, selecting feature indices solely by activation strength can assign the same token to distinct items.
Within a domain, shared dominant features can occupy the token budget, excluding less active features that distinguish individual items.
To allow these features to compete for inclusion, we form a candidate pool $\mathcal C_i$ containing the $M$ features with the highest positive activation strengths.
For each feature $j\in\mathcal C_i$, we define a \textit{selection priority} that favors strong activations while discounting features prevalent:
\begin{equation}
    \pi_{ij}
    =
    \log(1+a_{ij})\,\log\frac{1}{\rho_j}.
    \label{eq:pid-selection}
\end{equation}
where $\rho_j$ denotes the smoothed prevalence of feature $j$, estimated from its occurrence in candidate pools of training-visible items. 
Lower prevalence can compensate for weaker activation, allowing less active attributes to compete with shared dominant features. 
We then select the $K$ features with the highest priorities $\pi_{ij}$ and order their indices by decreasing activation strength $a_{ij}$. 
Let $c_{i,r}$ denote the feature index at position $r$. 
The resulting sequence $c_i=(c_{i,1},\ldots,c_{i,K})$ thus forms the token component of the PID. 

Although selection by priority mechanism mitigates collisions caused by shared dominant features, a compact token budget still leave semantically similar items with identical token sequences. 
We observe that the selected features can nevertheless exhibit different relative activation strengths across them.  
To preserve these differences, we retain the relative activation strengths in a vector $\mathbf{q}_i=(q_{i,1},\ldots,q_{i,K})$ aligned with the token sequence $c_i$. 
Each entry $q_{i,r}$ records normalized log strength of the feature indexed by $c_{i,r}$:
\begin{equation}
    q_{i,r}
    =
    \frac{\log(1+a_{i,c_{i,r}})}
    {\sum_{k=1}^{K}\log(1+a_{i,c_{i,k}})}.
    \label{eq:pid-profile}
\end{equation}
We retain the corresponding pretrained semantic concepts as $\mathbf m_i=(m_{c_{i,1}},\ldots,m_{c_{i,K}})$, associating each feature index with its human-readable meaning.
The feature-index sequence, relative activation strengths, and aligned semantic annotations together form the PID: 
\begin{equation}
    \mathrm{PID}_i
    =
    \bigl(c_i,\mathbf{q}_i,\mathbf{m}_i\bigr).
    \label{eq:pid-definition}
\end{equation}
By retaining relative activation strengths, PIDs can distinguish items that share the same feature-index sequence. 
The following proposition further bounds the number of such items.

\begin{proposition}[Item separability with PID]
\label{prop:pid-separation}
For a group of items sharing the same token sequence, let $\Delta$ denote the minimum pairwise $\ell_2$ distance between their mean PLM representations. 
If $\Delta>2\epsilon$, the maximum number $n_{\max}$ of items
sharing the same PID satisfies
\begin{equation}
    n_{\max}
    \leq
    1+\left\lfloor
    \frac{\delta}{\Delta-2\epsilon}
    \right\rfloor,
    \label{eq:pid-collision-bound}
\end{equation}
where $\epsilon$ bounds the reconstruction error using the selected features, and $\delta$ bounds the range of $\sum_{r=1}^{K}a_{i,c_{i,r}}$ across items in the group.
\end{proposition}

\noindent\textit{Proof.} See Appendix~\ref{app:pid-separation}. 
The bound relates collisions to reconstruction error and variation in summed strengths. 
When $\Delta>2\epsilon+\delta$,the strengths can distinguish items when their token sequence coincide.

\subsection{SPA: PID Prediction}
\label{sec:spa-prediction}

Although PIDs retain relative activation strengths alongside feature, predicting the indices alone does not explicitly model these strengths.
The Semantic Prediction Adapter (SPA) therefore incorporates feature-index and activation strengths into history encoding and predicts the next item's token sequence and relative strengths.

\textbf{Encoding history based PIDs.}
Though the PID construction selects $K$ features for compact item discrimination, the excluded may still provide evidence of user preferences. 
To retain this, SPA encodes each historical item using the candidate pool $\mathcal{C}_i$ from which $c_i$ is selected, together with the corresponding activation strengths. 
Specifically, we represent item $i$ by the feature--strength pairs $\mathcal{E}_i=\{(j,a_{ij}):j\in\mathcal{C}_i\}$. 
The feature encoder $\Psi_\theta$ combines trainable feature-index embeddings with activation-strength encodings and aggregates them into a fixed number of tokens per item. 
Given the interaction history $\mathcal{H}_u=(i_1,\ldots,i_T)$, a recommendation Transformer $F_\theta$ processes these tokens in interaction order:
\begin{equation}
    \mathbf{H}_u
    =
    F_\theta\!\left( 
        \Psi_\theta(\mathcal{E}_{i_1}),\ldots,
        \Psi_\theta(\mathcal{E}_{i_T})
    \right).
    \label{eq:spa-history}
\end{equation}

\textbf{Predicting target PIDs.}
Given the history representation $\mathbf H_u$, SPA models the next item's feature-index sequence through the conditional distribution $p_\theta(c\mid\mathbf H_u)$. 
SPA also predicts a strength vector $\hat{\mathbf q}_u$ from $\mathbf H_u$, where $\hat q_{u,r}$ estimates $q_{i^+,r}$ at activation-ranked position $r$. 
Its entries satisfy $\hat q_{u,r}\geq 0$ and $\sum_{r=1}^{K}\hat q_{u,r}=1$. 
The semantic annotations are obtained by looking up $m_j$ for each predicted feature index $j$. 
Its entries are normalized across the $K$ positions, satisfying $\sum_{r=1}^{K}\hat q_{u,r}=1$. 
In practice, the feature-index sequence can be generated autoregressively~\cite{geng2022recommendation} or in parallel~\cite{hou2025generating}. 
The remaining problem is how to favor the next item over competing catalog items.


\begin{table*}[t]
\centering
\caption{Overall recommendation performance. Best and second-best results are shown in \textbf{bold} and \underline{underlined}, respectively.}
\label{tab:main_results_wide}
\small
\setlength{\tabcolsep}{3pt}
\renewcommand{\arraystretch}{1.08}
\begin{tabular*}{\textwidth}{@{\extracolsep{\fill}}l*{12}{c}@{}}
\toprule
& \multicolumn{4}{c}{\textbf{Beauty}}
& \multicolumn{4}{c}{\textbf{Sports}}
& \multicolumn{4}{c}{\textbf{Toys}} \\
\cmidrule(lr){2-5}\cmidrule(lr){6-9}\cmidrule(l){10-13}
Method & \shortstack{R@5} & \shortstack{R@10} & \shortstack{N@5} & \shortstack{N@10} & \shortstack{R@5} & \shortstack{R@10} & \shortstack{N@5} & \shortstack{N@10} & \shortstack{R@5} & \shortstack{R@10} & \shortstack{N@5} & \shortstack{N@10} \\
\midrule
GRU4Rec
 & 0.0433 & 0.0646 & 0.0304 & \underline{0.0372}
  & 0.0222 & 0.0371 & 0.0141 & \underline{0.0189}
 & 0.0384 & 0.0543 & \underline{0.0272} & \underline{0.0324} \\
Bert4Rec
 & 0.0281 & 0.0434 & 0.0184 & 0.0233
 & 0.0129 & 0.0224 & 0.0079 & 0.0110
  & 0.0201 & 0.0313 & 0.0133 & 0.0168 \\
Tiger (2023)
 & \textbf{0.0464} & \underline{0.0662} & 0.0305 & 0.0369
 & \underline{0.0237} & 0.0341 & \underline{0.0156} & \underline{0.0189}
 & \underline{0.0402} & 0.0584 & 0.0263 & 0.0322 \\
LCRec (2024)
 & 0.0228 & 0.0362 & 0.0159 & 0.0197
 & 0.0082 & 0.0133 & 0.0073 & 0.0089
 & 0.0174 & 0.0285 & 0.0119 & 0.0157 \\
RPG (2025)
 & 0.0322 & 0.0502 & 0.0210 & 0.0269
 & 0.0176 & 0.0266 & 0.0124 & 0.0153
 & 0.0306 & 0.0477 & 0.0192 & 0.0247 \\
UTGRec (2026)
& 0.0193	& 0.0330 &	0.0122	& 0.0167	 & 0.0126 &	0.0219	& 0.0080 &	0.0109	& 0.0177 &	0.0304 &	0.0107	&0.0147 \\

Ghost (2026)
 & 0.0388 & 0.0590 & 0.0257 & 0.0322
 & 0.0260 & \textbf{0.0422} & 0.0171 & 0.0223
 & 0.0333 & 0.0543 & 0.021 & 0.0277 \\
DIGR (2026)
 & 0.0431 & 0.0605 & \textbf{0.0307} & 0.0362
 & 0.0186 & 0.0291 & 0.0129 & 0.0161
 & 0.0396 & 0.0571 & 0.0268 & 0.0319 \\
TDPM (2026)
 & 0.0418 & 0.0655 & 0.0279 & 0.0355
 & 0.0216 & 0.0321 & 0.0140 & 0.0173
 & 0.0398 & \underline{0.0591} & 0.0259 & 0.0322 \\
\midrule
\textbf{API (Ours)}
 & \underline{0.0457} & \textbf{0.0694} & \underline{0.0305} & \textbf{0.0389}
 & \textbf{0.0271} & \underline{0.0411} & \textbf{0.0188} & \textbf{0.0239}
 & \textbf{0.0414} & \textbf{0.0616} & \textbf{0.0291} & \textbf{0.0353} \\
\bottomrule
\end{tabular*}
\vspace{-0.1cm}
\end{table*}
\subsection{SPA: Item-Level Supervision}
\label{sec:spa-competition}

Building on the feature-index and strength predictions, we use the observed next item's PID to construct item-level supervision for GR training. 
We combine sequence likelihood and strength agreement into the score to supervise PID prediction and candidates ranking.

\textbf{Comparing candidate PIDs.}
For a item $i$, let $\ell_u(i)=\log p_\theta(c_i\mid\mathbf{H}_u)$ denote the log likelihood of its feature-index sequence. 
We then define the item score by balancing sequence likelihood $\ell_u(i)$ with agreement between the predicted and candidate relative strengths:
\begin{equation}
s_\theta(u,i)=\ell_u(i)-\beta D_{\mathrm{H}}(\hat{\mathbf{q}}_u,\mathbf{q}_i),
\label{eq:spa-item-score}
\end{equation}
where $\beta\geq 0$ controls the strength penalty. 
The squared Hellinger distance $D_{\mathrm{H}}$ compares the predicted strengths $\hat{\mathbf{q}}_u$ with the candidate's strengths $\mathbf{q}_i$ at corresponding activation-ranked positions:
\begin{equation}
D_{\mathrm{H}}(\hat{\mathbf{q}}_u,\mathbf{q}_i)=1-\sum_{r=1}^{K}\sqrt{\hat q_{u,r}q_{i,r}}.
\label{eq:spa-strength-distance}
\end{equation}
For these normalized vectors, the distance lies in $[0,1]$
and equals zero exactly when their relative strengths match.

\textbf{Learning user preferences.}
For the observed next item $i^+$, minimizing $-s_\theta(u,i^+)$ supervises its feature-index sequence and relative strengths. 
However, a better match to the target PID still leave competing items with higher scores. 
To encourage the target item to outrank competing candidates, we introduce a preference alignment loss over the candidate set $\mathcal{J}_u\subseteq\mathcal{I}$ containing $i^+$:
\begin{equation}
\mathcal{L}_{\mathrm{pref}}(u,i^+)=\log\sum_{n\in\mathcal{J}_u}\exp\big(s_\theta(u,n)-s_\theta(u,i^+)\big).
\label{eq:spa-catalog-loss}
\end{equation}
We then minimize the joint loss over observed interactions:
\begin{equation}
\mathcal{L}=\mathbb{E}_{(u,i^+)}\big[-s_\theta(u,i^+)+\lambda\mathcal{L}_{\mathrm{pref}}(u,i^+)\big],
\label{eq:spa-objective}
\end{equation}
where $\lambda\geq 0$ weights preference alignment. 
The following proposition shows how the relative strengths in PIDs can help SPA rank the target above the competing items with the same feature sequence.

\begin{proposition}[Item ranking with SPA]
\label{prop:spa-strength-margin}
For a target item $i^+$ and a competing item $i^-$ sharing the same feature-index sequence, let $\Delta_q=D_{\mathrm{H}}(\mathbf{q}_{i^+},\mathbf{q}_{i^-}) > 0$ denote their strength separation.
For $\beta>0$, the target's ranking margin satisfies
\begin{equation}
s_\theta(u,i^+)-s_\theta(u,i^-)\geq\beta\big(\Delta_q-2\sqrt{\Delta_q\varepsilon_u}\big).
\label{eq:spa-strength-margin}
\end{equation}
where $\hat{\mathbf{q}}_u$ denotes SPA's prediction of the next item's relative strengths given $\mathbf{H}_u$, and $\varepsilon_u=D_{\mathrm{H}}(\hat{\mathbf{q}}_u,\mathbf{q}_{i^+})$ measures its error against the target PID.
In particular, $\varepsilon_u<\Delta_q/4$ ensures that SPA ranks $i^+$ above $i^-$.
\end{proposition}

\noindent\textit{Proof.} See Appendix~\ref{app:spa-strength-margin}.
Under the stated condition, the strength differences retained in PIDs yield a positive ranking margin for the target. 
At retrieval, SPA ranks items using the same score $s_\theta(u,i)$.

\subsection{Retrieval Explanation}
\label{sec:retrieval-explanation}

To explain why SPA favors one candidate over another, we decompose their margin into contributions from individual interpretable features. 
We then associate these contributions with their semantic annotations to explain the feature evidence behind the ranking. 

\textbf{Decomposing retrieval scores to pretrained features.} 
To quantify each feature's contribution to the item score, we decompose the sequence log likelihood as $\ell_u(i)=\sum_{r=1}^{K}\ell_{u,r}(i)$.
Here, $\ell_{u,r}(i)$ denotes the log probability assigned by SPA to feature index $c_{i,r}$ at position $r$. 
Combining this term with the strength agreement at the same position gives the feature contribution:
\begin{equation}
\phi_{u,i}(c_{i,r})=\ell_{u,r}(i)+\beta\sqrt{\hat q_{u,r}q_{i,r}},
\label{eq:retrieval-feature-contribution}
\end{equation}
where $\phi_{u,i}(c_{i,r})$ measures the contribution of feature $c_{i,r}$ to candidate $i$'s score.
For features $j$ absent from $c_i$, we set $\phi_{u,i}(j)=0$.

\textbf{Explaining candidate rankings.}
To explain the ranking of candidates $i$ and $i'$, we compare the contributions of each feature $j$ to their respective scores. 
Their score difference decomposes as:
\begin{equation}
s_\theta(u,i)-s_\theta(u,i')=\sum_j\left(\phi_{u,i}(j)-\phi_{u,i'}(j)\right).
\label{eq:retrieval-score-difference}
\end{equation}
Each term measures feature $j$'s contribution to this difference, with positive values favoring $i$ and negative values favoring $i'$.
We pair these contributions with the pretrained semantic annotations $m_j$ to show which features favor each candidate and by how much.
\section{Experiments}\label{sec:experiment}
\providecommand{\ExpPending}[1]{\par\noindent\textit{[To complete: #1]}\par}

\begin{table*}[t]
\centering
\caption{
Recommendation performance on three datasets with T5 and GPT-2 backbones.
$K$ counts semantic tokens.
R@$k$ and N@$k$ denote Recall@$k$ and NDCG@$k$.
The PID with the SPA constitutes API.
For each backbone and decoding strategy, \textbf{bold} and
\underline{underlined} values indicate the best and second-best results,
respectively, including ties.}
\label{tab:main_results}
\begingroup
\footnotesize
\setlength{\tabcolsep}{2pt}
\renewcommand{\arraystretch}{1.10}
\begin{tabularx}{\textwidth}{@{}ccc*{12}{>{\centering\arraybackslash}X}@{}}
\toprule
\multirow{2}{*}{\textbf{Back.}} &
\multirow{2}{*}{\textbf{Generation}} &
\multirow{2}{*}{\textbf{Tokenization}} &
\multicolumn{4}{c}{\textbf{Beauty}} &
\multicolumn{4}{c}{\textbf{Sports}} &
\multicolumn{4}{c}{\textbf{Toys}} \\
\cmidrule(lr){4-7}\cmidrule(lr){8-11}\cmidrule(lr){12-15}
 & & & R@5 & R@10 & N@5 & N@10 & R@5 & R@10 & N@5 & N@10 & R@5 & R@10 & N@5 & N@10 \\
\midrule
\multirow{16}{*}{\textbf{T5}} & \multirow{6}{*}{\textbf{AR}}
& \textbf{Random} & 0.0310 & 0.0454 & 0.0216 & 0.0262 & 0.0082 & 0.0111 & 0.0051 & 0.0060 & 0.0188 & 0.0250 & 0.0140 & 0.0161 \\
& & \textbf{RVQ} & \underline{0.0464} & 0.0662 & \underline{0.0305} & 0.0369 & \underline{0.0237} & 0.0341 & \underline{0.0156} & \underline{0.0189} & \underline{0.0402} & 0.0584 & 0.0263 & 0.0322 \\
& & \textbf{RQVAE} & 0.0402 & 0.0602 & 0.0269 & 0.0334 & 0.0215 & 0.0310 & 0.0139 & 0.0170 & 0.0344 & 0.0492 & 0.0240 & 0.0290 \\
& & \textbf{RKMeans} & 0.0441 & \underline{0.0667} & \underline{0.0305} & \underline{0.0378} & 0.0217 & 0.0330 & 0.0145 & 0.0181 & 0.0384 & 0.0564 & 0.0256 & 0.0314 \\
& & \textbf{HiD-VAE} & 0.0419 & 0.0644 & 0.0283 & 0.0355 & 0.0215 & 0.0328 & 0.0138 & 0.0175 & 0.0388 & \underline{0.0595} & 0.0257 & 0.0324 \\
& & \textbf{PID} & 0.0435 & 0.0630 & 0.0299 & 0.0361 & 0.0210 & \underline{0.0346} & 0.0135 & 0.0178 & 0.0401 & 0.0585 & \underline{0.0269} & \underline{0.0328} \\
\cmidrule(lr){2-15}
& \multirow{1}{*}{\textbf{+ SPA}}
& \textbf{PID} & \textbf{0.0476} & \textbf{0.0701} & \textbf{0.0323} & \textbf{0.0396} & \textbf{0.0243} & \textbf{0.0372} & \textbf{0.0157} & \textbf{0.0198} & \textbf{0.0419} & \textbf{0.0607} & \textbf{0.0287} & \textbf{0.0347} \\
\cmidrule(lr){2-15}
& \multirow{6}{*}{\textbf{PR}} & \textbf{OPQ ($K=4$)} & 0.0290 & 0.0466 & 0.0196 & 0.0256 & 0.0057 & 0.0105 & 0.0034 & 0.0049 & 0.0248 & 0.0373 & 0.0160 & 0.0200 \\
& & \textbf{OPQ ($K=8$)} & 0.0368 & 0.0543 & 0.0247 & 0.0303 & 0.0080 & 0.0134 & 0.0052 & 0.0070 & 0.0320 & 0.0485 & 0.0210 & 0.0263 \\
& & \textbf{OPQ ($K=16$)} & 0.0312 & 0.0484 & 0.0203 & 0.0259 & 0.0114 & 0.0199 & 0.0076 & 0.0102 & 0.0278 & 0.0424 & 0.0171 & 0.0218 \\
& & \textbf{PID ($K=4$)} & 0.0278 & 0.0402 & 0.0194 & 0.0234 & 0.0049 & 0.0087 & 0.0029 & 0.0041 & 0.0294 & 0.0488 & 0.0182 & 0.0245 \\
& & \textbf{PID ($K=8$)} & 0.0370 & 0.0541 & 0.0262 & 0.0316 & 0.0094 & 0.0158 & 0.0062 & 0.0083 & 0.0332 & \textbf{0.0547} & 0.0207 & 0.0276 \\
& & \textbf{PID ($K=16$)} & 0.0340 & 0.0523 & 0.0225 & 0.0284 & 0.0132 & 0.0194 & 0.0086 & 0.0106 & 0.0348 & 0.0509 & 0.0209 & 0.0262 \\
\cmidrule(lr){2-15}
& \multirow{3}{*}{\textbf{+ SPA}}
& \textbf{PID ($K=4$)} & 0.0388 & 0.0572 & 0.0264 & 0.0323 & \underline{0.0168} & \underline{0.0267} & \underline{0.0111} & \underline{0.0142} & \underline{0.0352} & \underline{0.0517} & \textbf{0.0242} & \textbf{0.0295} \\
& & \textbf{PID ($K=8$)} & \textbf{0.0401} & \textbf{0.0627} & \textbf{0.0274} & \textbf{0.0346} & \textbf{0.0192} & \textbf{0.0289} & \textbf{0.0126} & \textbf{0.0157} & \textbf{0.0354} & \underline{0.0517} & \underline{0.0229} & \underline{0.0282} \\
& & \textbf{PID ($K=16$)} & \underline{0.0394} & \underline{0.0575} & \underline{0.0271} & \underline{0.0329} & 0.0162 & 0.0241 & 0.0106 & 0.0131 & 0.0319 & 0.0462 & 0.0210 & 0.0256 \\
\midrule
\multirow{16}{*}{\textbf{GPT2}} & \multirow{6}{*}{\textbf{AR}}
& \textbf{Random} & 0.0292 & 0.0401 & 0.0207 & 0.0247 & 0.0070 & 0.0122 & 0.0039 & 0.0056 & 0.0133 & 0.0176 & 0.0094 & 0.0108 \\
& & \textbf{RVQ} & 0.0394 & 0.0594 & 0.0255 & 0.0320 & 0.0197 & 0.0331 & 0.0126 & 0.0169 & 0.0327 & 0.0536 & 0.0206 & 0.0274 \\
& & \textbf{RQVAE} & 0.0362 & 0.0542 & 0.0242 & 0.0301 & 0.0166 & 0.0271 & 0.0103 & 0.0137 & 0.0296 & 0.0439 & 0.0192 & 0.0238 \\
& & \textbf{RKMeans} & \underline{0.0423} & 0.0640 & 0.0278 & 0.0348 & 0.0182 & 0.0294 & 0.0110 & 0.0145 & 0.0313 & 0.0508 & 0.0195 & 0.0258 \\
& & \textbf{HiD-VAE} & 0.0398 & 0.0607 & 0.0261 & 0.0329 & \underline{0.0223} & \underline{0.0360} & \underline{0.0141} & \underline{0.0185} & 0.0335 & \underline{0.0567} & 0.0211 & 0.0286 \\
& & \textbf{PID} & 0.0420 & \underline{0.0663} & \underline{0.0282} & \underline{0.0361} & 0.0190 & 0.0292 & 0.0126 & 0.0159 & \underline{0.0380} & 0.0529 & \underline{0.0265} & \underline{0.0313} \\
\cmidrule(lr){2-15}
& \multirow{1}{*}{\textbf{+ SPA}}
& \textbf{PID} & \textbf{0.0457} & \textbf{0.0694} & \textbf{0.0305} & \textbf{0.0389} & \textbf{0.0271} & \textbf{0.0411} & \textbf{0.0188} & \textbf{0.0239} & \textbf{0.0414} & \textbf{0.0616} & \textbf{0.0291} & \textbf{0.0353} \\
\cmidrule(lr){2-15}
& \multirow{6}{*}{\textbf{PR}}
& \textbf{OPQ ($K=4$)} & 0.0292 & 0.0457 & 0.0197 & 0.0250 & 0.0100 & 0.0175 & 0.0065 & 0.0089 & 0.0267 & 0.0449 & 0.0164 & 0.0222 \\
& & \textbf{OPQ ($K=8$)} & 0.0374 & 0.0549 & 0.0248 & 0.0305 & 0.0161 & 0.0241 & 0.0109 & 0.0135 & 0.0274 & 0.0441 & 0.0165 & 0.0219 \\
& & \textbf{OPQ ($K=16$)} & 0.0322 & 0.0502 & 0.0210 & 0.0269 & 0.0176 & 0.0266 & \underline{0.0124} & 0.0153 & 0.0306 & 0.0477 & 0.0192 & 0.0247 \\
& & \textbf{PID ($K=4$)} & 0.0267 & 0.0393 & 0.0193 & 0.0234 & 0.0153 & 0.0234 & 0.0101 & 0.0127 & 0.0330 & 0.0483 & 0.0230 & 0.0279 \\
& & \textbf{PID ($K=8$)} & 0.0379 & 0.0540 & 0.0270 & 0.0322 & 0.0162 & \textbf{0.0309} & 0.0118 & \textbf{0.0167} & 0.0307 & 0.0447 & 0.0214 & 0.0259 \\
& & \textbf{PID ($K=16$)} & 0.0336 & 0.0525 & 0.0228 & 0.0288 & 0.0179 & 0.0281 & 0.0121 & 0.0156 & 0.0267 & 0.0383 & 0.0189 & 0.0226 \\
\cmidrule(lr){2-15}
& \multirow{3}{*}{\textbf{+ SPA}}
& \textbf{PID ($K=4$)} & \textbf{0.0411} & \textbf{0.0618} & \textbf{0.0277} & \textbf{0.0343} & \textbf{0.0195} & \underline{0.0300} & \textbf{0.0129} & \underline{0.0163} & \textbf{0.0393} & \textbf{0.0559} & \textbf{0.0281} & \textbf{0.0334} \\
& & \textbf{PID ($K=8$)} & 0.0395 & \underline{0.0573} & \underline{0.0274} & 0.0331 & \underline{0.0186} & 0.0294 & 0.0120 & 0.0154 & \underline{0.0384} & \underline{0.0531} & \underline{0.0271} & \underline{0.0318} \\
& & \textbf{PID ($K=16$)} & \underline{0.0398} & \underline{0.0573} & \textbf{0.0277} & \underline{0.0333} & 0.0161 & 0.0253 & 0.0104 & 0.0133 & 0.0377 & 0.0515 & 0.0262 & 0.0306 \\
\bottomrule
\end{tabularx}
\endgroup
\end{table*}

We evaluate pretrained vocabulary reuse and proposed API to generative recommendation through the following research questions:
\begin{itemize}[leftmargin=*]

\item  \textbf{RQ1:} How well does API recommend across decoding strategies? 
\item  \textbf{RQ2:} How well can PIDs preserve item distinctions and content? 
\item  \textbf{RQ3:} Which supervision signals improve item retrieval? 
\item  \textbf{RQ4:} How effectively can the pretrained vocabulary be reused? 
\item  \textbf{RQ5:} Can item retain readable meanings with the PID construction and GR model support targeted-feature steering to retrieval?
\end{itemize}

\vspace{-0.2cm}
\subsection{Experimental Setup} \label{sec:5.1}

\noindent\textbf{Benchmark.} 
We use the Beauty, Sports and Toys subsets of the preprocessed Amazon datasets aligned with prior work~\cite{geng2022recommendation,ju2025grid}, retaining their dataset splits and evaluation protocol. 
Item descriptions combine the title, brand, category and price fields. 
Recommendation models are trained on training interactions; PID selection statistics and activation calibration are estimated from training-item contents. 
Dataset statistics and details are provided in Appendix~\ref{app:experimental_details}.

\textbf{Compared methods.} 
We organize the main comparisons at two levels: item tokenization in Table~\ref{tab:main_results} and complete recommendation frameworks in Table~\ref{tab:main_results_wide}. 
The tokenization comparison includes Random, OPQ, HiD-VAE and residual quantization baselines (RVQ, RQ-VAE and RKMeans) that aligned with GRID~\cite{ju2025grid}. 
We then evaluate them with T5- and GPT-2-style backbones, adopting AR and PR implementations based on TIGER~\cite{rajput2023recommender} and RPG~\cite{hou2025generating}, respectively. 
Detailed baseline and configurations are provided in Appendix~\ref{app:baselines}. 

\textbf{Implementation.} 
We extract item features using pretrained GPT-2-small~\cite{radford2019language} and its paired P-SAE with $24,576$ labeled dictionary features, keeping both models frozen. 
The T5-style and GPT-2-style recommendation backbones are trained from random initialization. 
We retain 32 candidate features per item and use $K\in\{4,8,16\}$ tokens. 
We optimize the recommendation models using Adam with a learning rate of $10^{-3}$ and weight decay of $10^{-4}$. 
The loss and strength weights are set to $\lambda_{\mathrm{pref}}=0.1$ and $\beta=8$, respectively. 
Checkpoints are selected by Recall@5 on the validation split.
Detailed configurations are provided in Appendix~\ref{app:implementation}.

\textbf{Evaluation.} 
We report Recall@5/10 and NDCG@5/10 on the complete test split, using the full catalog without sampled negatives. 
For PID discrimination, \textit{collision rate} measures the fraction of items whose selected feature set is shared by another item. 
For content preservation, \textit{Overlap@20} measures the average fraction of shared neighbors between representation-based and reference top-20 lists. 
For steering, we report \textit{feature exposure@10}, the fraction of top-10 recommendations whose PIDs contain the target feature. 
PR uses exact catalog scoring, while AR uses catalog-constrained beam search with a width of 50.
Detailed are provided in Appendix~\ref{app:evaluation}.

\begin{findingsbox}{Findings for RQ1 (Recommendation Performance)}
\noindent\textbf{Finding 1: Effective vocabulary reuse.}
Pretrained features support competitive GR without additional vocabulary learning.

\noindent\textbf{Finding 2: Benefits of recommendation adaptation.}
SPA improves over direct PID in all evaluated AR settings and most PR settings.

\noindent\textbf{Finding 3: Competitive framework performance.}
API achieves the best reported results in nine of twelve dataset--metric comparisons.
\end{findingsbox}

\subsection{Recommendation Performance (RQ1)}\label{sec:5.2}

Table~\ref{tab:main_results_wide} compares the complete API framework with sequential and generative recommenders. 
\textbf{API achieves the best reported result in nine of the twelve dataset--metric comparisons}, including all four metrics on Toys.
It ranks second on Beauty Recall@5/NDCG@5 and Sports Recall@10, with a tie on Beauty NDCG@5.
Table~\ref{tab:main_results} further examines how recommendation adaptation affects the use of the pretrained vocabulary.
Direct PID does not uniformly outperform learned tokenizers: on Beauty with T5-PR and $K=4$, its Recall@5 is 0.0278, compared with 0.0290 for OPQ, while adding SPA raises it to 0.0388.
Across the two backbones and three datasets, \textbf{SPA improves Recall@5 over direct PID in all six AR settings and 16 of the 18 PR settings}.
With T5-AR, the complete API pipeline exceeds the strongest learned-tokenizer baseline by 2.5\%--4.2\% in Recall@5.
Both PR exceptions occur at $K=16$, while the best API-PR Recall@5 for every backbone--dataset pair is attained at $K=4$ or $K=8$.
These results also show that \textbf{increasing identifier length does not consistently improve recommendation}.

\begin{figure*}[t]
    \centering
    \includegraphics[width=1\linewidth]{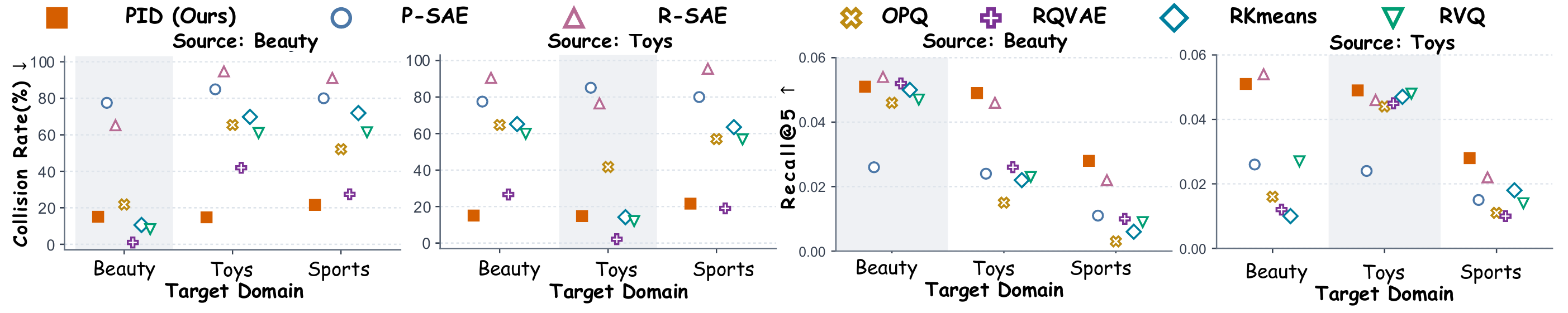}
\caption{Vocabulary reuse across domains. Collision rates and downstream
Recall@5 are reported for each source--target pair; shaded columns mark the
source domain of data-trained tokenizers.}
    \label{fig:vocabulary_reuse}
    \vspace{-0.3cm}
\end{figure*}

\subsection{PID Quality (RQ2)}\label{sec:5.3}

Table~\ref{tab:pid_collision} compares feature selection policies within the same top-32 candidate pool, measuring collisions between selected feature sets.
\textbf{PID improves item discrimination under compact budgets}: at $K=4$, it reduces the mean collision rate across the three datasets from 80.76\% under native selection to 17.09\%.
The policies coincide at $K=32$, when the entire candidate pool is retained.
Random selection yields even fewer collisions, but Figure~5\ref{fig:pid_content} shows substantially weaker content alignment at $K=4$ and $K=8$.
At these budgets, PID achieves the highest Overlap@20 among the compared methods on Beauty and Sports against both title-based and PLM-based references.
These comparisons indicate that \textbf{low collision rates alone do not establish content preservation}.
To further separate these properties, Table~\ref{tab:sae-profile-permutation} permutes complete feature--activation profiles across items while keeping item content fixed.
Collision rates remain unchanged, whereas Overlap@20 falls below 0.2\% against both references.
Finally, \textbf{activation-aware neighborhoods show stronger content alignment}: with the selected feature indices fixed at $K=8$, activation-weighted cosine neighborhoods exceed index-only Jaccard neighborhoods by 1.03\%, 2.60\% and 4.41 \% in PLM-based Overlap@20 on all benchmarks.

\begin{table}[t]
\centering
\small
\setlength{\tabcolsep}{2pt}
\renewcommand{\arraystretch}{0.95}
\setlength{\aboverulesep}{1pt}
\setlength{\belowrulesep}{1pt}
\captionsetup{skip=3pt}
\caption{Collision-affected items (\%) across token budgets.}
\label{tab:pid_collision}
\begin{tabular*}{\linewidth}{@{\extracolsep{\fill}}llrrrrr@{}}
\toprule
Dataset & Selection & $K=4$ & $K=8$ & $K=16$ & $K=24$ & $K=32$ \\
\midrule
\multirow{3}{*}{Beauty}
 & Native SAE & 77.37 & 42.72 & 11.11 & 2.86 & 1.23 \\
 & Random ID  &  2.50 &  0.00 &  0.00 & 0.00 & 1.23 \\
 & \textbf{PID} & 15.16 & 3.33 & 1.26 & 1.19 & 1.23 \\
\midrule[0.3pt]
\multirow{3}{*}{Sports}
 & Native SAE & 79.96 & 51.77 & 19.98 & 6.30 & 3.44 \\
 & Random ID  &  3.67 &  0.00 &  0.00 & 0.00 & 3.44 \\
 & \textbf{PID} & 21.48 & 7.71 & 3.43 & 3.20 & 3.44 \\
\midrule[0.3pt]
\multirow{3}{*}{Toys}
 & Native SAE & 84.95 & 57.59 & 17.96 & 4.68 & 1.49 \\
 & Random ID  &  2.59 &  0.00 &  0.00 & 0.00 & 1.49 \\
 & \textbf{PID} & 14.64 & 4.54 & 1.51 & 1.49 & 1.49 \\
\bottomrule
\end{tabular*}
\vspace{-0.3cm}
\end{table}

\begin{findingsbox}{Findings for RQ2 (PID Quality)}
\noindent\textbf{Finding 1: Compact discrimination.}
Priority-based feature selection reduces collisions within the same frozen vocabulary.

\noindent\textbf{Finding 2:} Discrimination do not guarantee content alignment.
\end{findingsbox}

\begin{figure*}[ht]
    \centering
    \includegraphics[width=1\linewidth]{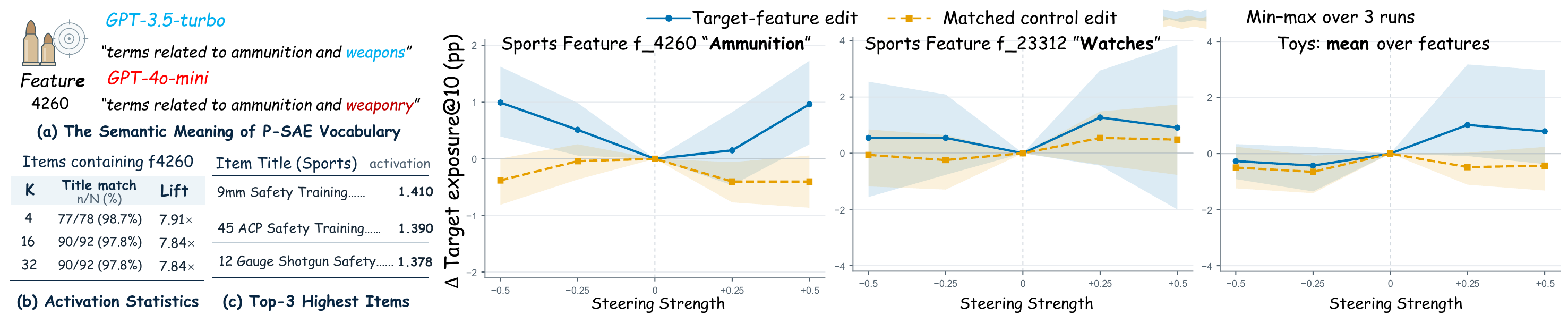}
\caption{\textbf{Interpretability and controlled steering with P-SAE features.}
Left: The labeled $f_{4260}$ by the LLMs, title matches to the meanings and top-activating items. 
$K$ is the PID budget; lift uses the catalog baseline ($\approx12.5\%$). 
Right: changes in target items exposure@10 under activation steering,
relative to the non-target. 
Shaded bands show the min--max range over three runs.}
    \label{fig:Interpretability}
    \vspace{-0.2cm}
\end{figure*}

\begin{figure}[t]
    \centering
    \includegraphics[width=1\columnwidth]{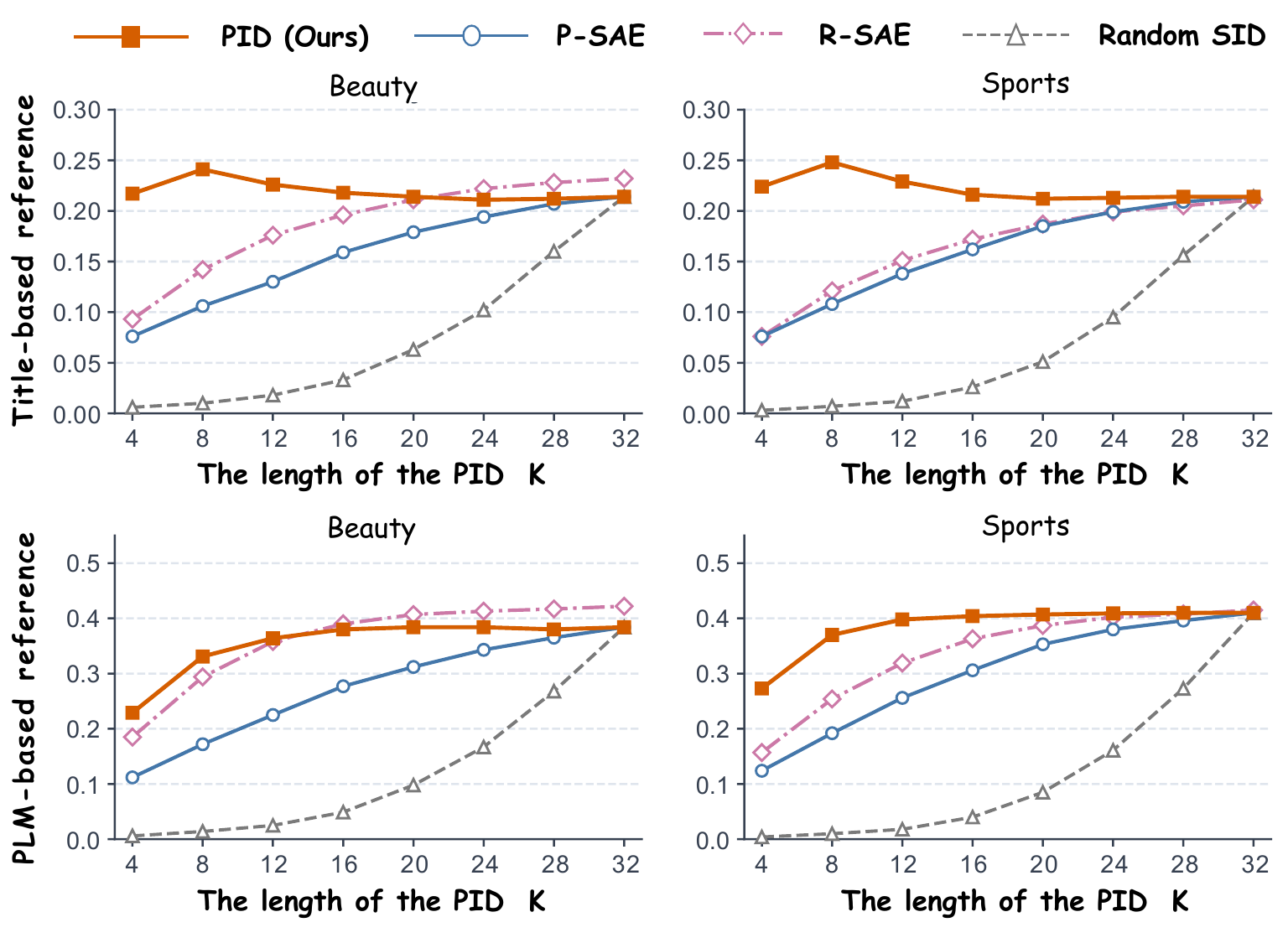}
    \vspace{-0.2cm}
\caption{Content alignment across semantic-token budgets, measured by nearest-neighbor overlap@20 against title-based and PLM-based references. 
Higher values indicate greater agreement with the reference neighborhoods.}
    \label{fig:pid_content}
    \vspace{-0.2cm}
\end{figure}

\subsection{SPA Supervision (RQ3)}
\label{sec:pid_quality}

Table~\ref{tab:spa_ablation} examines preference alignment and strength modeling under T5-PR with $K=8$, fixing the SPA architecture, activation-aware history encoding, feature-index targets, and output ordering.
Starting from SID prediction alone, we add preference alignment, strength prediction and scoring, or both.
\textbf{Preference alignment consistently improves retrieval}, increasing Recall@10 by 13.3\%, 45.9\%, and 9.4\% on Beauty, Sports, and Toys, respectively.
Adding strength prediction and scoring alone lowers NDCG@10 on all three datasets.
Once preference alignment is included, however, strength modeling further improves NDCG@10 by 4.0\%--6.0\%.
Full SPA achieves the highest Recall@10 and NDCG@10 throughout this ablation.
These results suggest that \textbf{strength evidence is more effective when coupled with item-level preference alignment}.

\begin{findingsbox}{Findings for RQ3}
\textbf{Finding 1: Learning through item competition.}
Catalog-grounded supervision provides the most consistent gains in this ablation.

\par\smallskip
\textbf{Finding 2: Complementary strength supervision.}
SPA provides additional ranking gains when combined with preference alignment. 
\end{findingsbox}


\begin{table}[t]
\centering
\small
\setlength{\tabcolsep}{1pt}
\renewcommand{\arraystretch}{1.05}
\caption{Content alignment under collision-preserving permutation
(\%). Collision rates (Coll.) remain unchanged.}
\label{tab:sae-profile-permutation}
\begin{tabular*}{\columnwidth}{
    @{\extracolsep{\fill}}l*{9}{r}@{}
}
\toprule
 & & \multicolumn{4}{c}{Index only}
   & \multicolumn{4}{c}{Index + Activation} \\
\cmidrule(lr){3-6}\cmidrule(lr){7-10}
 & & \multicolumn{2}{c}{Title}
   & \multicolumn{2}{c}{PLM}
   & \multicolumn{2}{c}{Title}
   & \multicolumn{2}{c}{PLM} \\
\cmidrule(lr){3-4}\cmidrule(lr){5-6}
\cmidrule(lr){7-8}\cmidrule(lr){9-10}
Dataset & Coll.
 & Orig. & Perm. & Orig. & Perm.
 & Orig. & Perm. & Orig. & Perm. \\
\midrule
Beauty & 3.33
 & 24.15 & 0.19 & 33.02 & 0.16
 & 25.54 & 0.17 & 34.05 & 0.14 \\
Sports & 7.71
 & 24.83 & 0.12 & 37.01 & 0.16
 & 26.21 & 0.11 & 39.61 & 0.15 \\
Toys & 4.54
 & 25.90 & 0.12 & 38.81 & 0.12
 & 27.47 & 0.15 & 43.22 & 0.14 \\
\bottomrule
\end{tabular*}
\vspace{-0.3cm}   
\end{table}

\vspace{-0.2cm}
\subsection{Vocabulary Reuse (RQ4)}
Figure~\ref{fig:vocabulary_reuse} evaluates learned tokenizers fitted on Beauty or Toys and frozen across all three target catalogs.
PID retains the frozen P-SAE dictionary and estimates IDF from target training-item contents; all recommenders are trained on target-domain interactions.
On Beauty and Toys, several learned tokenizers incur more collisions and lower Recall@5 when fitted on the other domain.
\textbf{PID reuses the vocabulary without refitting}, maintaining comparatively low collision rates and competitive Recall@5 across targets: it leads on Toys and Sports under both source configurations, while R-SAE remains ahead on Beauty.
Native top-$K$ selection produces substantially more collisions than PID using the same dictionary, highlighting \textbf{the importance of catalog-specific selection}.
Figure~\ref{fig:retrieval_explanation}(b) extrapolates single-task timings to three equal-workload tokenization tasks, counting UTGRec pretraining once and applicable fitting and processing per task.
Treating the pretrained PLM and P-SAE as available resources, cumulative time relative to one GPT-2 encoding task is $3.78\times$ for API, $407.69\times$ for RVQ, and $674.99\times$ for UTGRec.
\textbf{API thus avoids additional vocabulary-fitting costs}, while retaining per-task feature extraction and encoding.

\begin{table}[t]
\centering
\footnotesize
\setlength{\tabcolsep}{2pt}
\renewcommand{\arraystretch}{1.10}
\caption{
Supervision ablation of SPA under T5-style parallel prediction
with $K=8$ features.
}
\label{tab:spa_ablation}
\begin{tabular*}{\columnwidth}{
    @{\extracolsep{\fill}}lrrrrrr@{}
}
\toprule
&
\multicolumn{2}{c}{Beauty} &
\multicolumn{2}{c}{Sports} &
\multicolumn{2}{c}{Toys} \\
\cmidrule(lr){2-3}
\cmidrule(lr){4-5}
\cmidrule(lr){6-7}
Supervision
& R@10 & N@10
& R@10 & N@10
& R@10 & N@10 \\
\midrule

Index-only
& 0.05080 & 0.02959
& 0.01980 & 0.01104
& 0.04523 & 0.02518 \\

+ Strength
& 0.04905 & 0.02833
& 0.01893 & 0.01055
& 0.04533 & 0.02447 \\

+ Preference
& 0.05755 & 0.03312
& 0.02888 & 0.01509
& 0.04950 & 0.02656 \\

Full SPA
& \textbf{0.06265} & \textbf{0.03456}
& \textbf{0.02891} & \textbf{0.01569}
& \textbf{0.05172} & \textbf{0.02816} \\

\bottomrule
\end{tabular*}
\vspace{-0.2cm}
\end{table}

\begin{figure}
    \centering
    \includegraphics[width=1.0\linewidth]{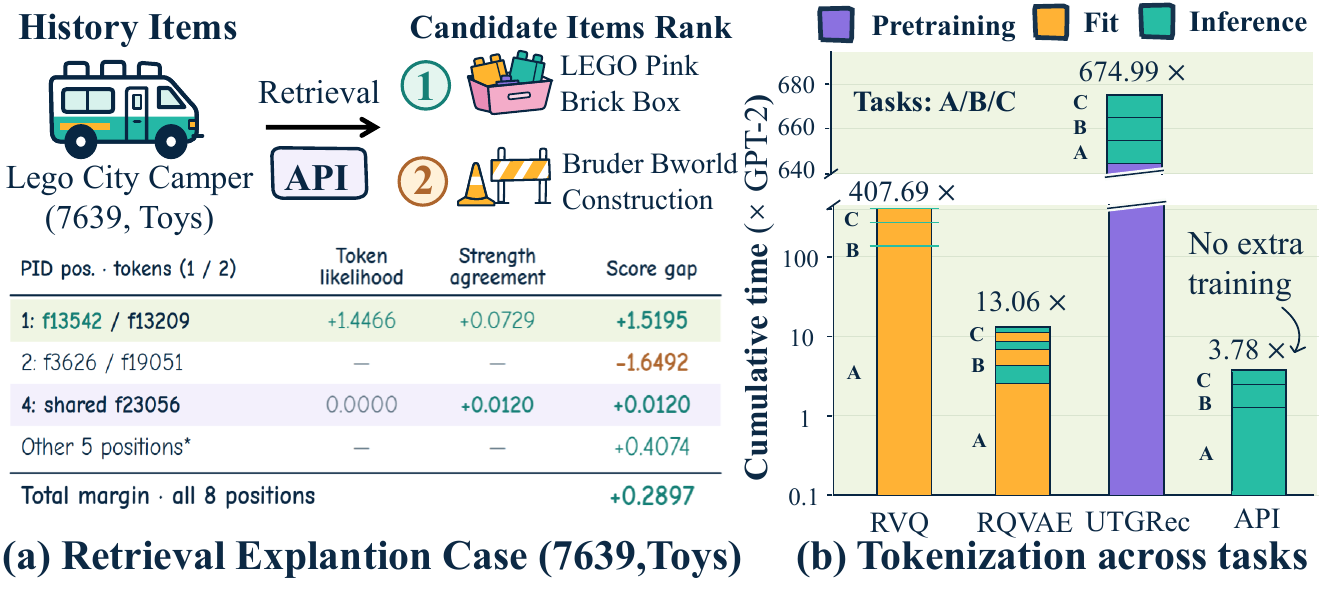}
\caption{\textbf{Retrieval explanation and tokenization efficiency.}
(a) A Toys case showing candidate score-gap contributions grouped
by PID position. 
(b) Projected umulative tokenization time for three
tasks, counting UTGRec pretraining once.
}
    \vspace{-0.3cm}
    \label{fig:retrieval_explanation}
\end{figure}

\vspace{-0.2cm}
\subsection{Retrieval Explanation (RQ5)}\label{sec:retrieval_explanation_experiment}

Figure~\ref{fig:Interpretability} illustrates \textbf{readable feature semantics} through Sports feature $f_{4260}$, whose external descriptions refer to ammunition and weapons. 
Among training items whose $K=4$ PIDs contain this feature, 77 of 78 match the hunting-and-shooting lexical rule, yielding 98.7\% agreement and a $7.91\times$ lift over the rule's prevalence in the training catalog. 
The highest-activation titles also describe training ammunition.
The Toys case in Figure~\ref{fig:retrieval_explanation} further provides \textbf{traceable ranking evidence}: a LEGO brick box outranks Bruder construction accessories by a margin of $+0.2897$, comprising $+0.2993$ from token likelihoods and $-0.0096$ from strength agreement.
At PID position~1, the contrast between LEGO-related $f_{13542}$ and $f_{13209}$ favors the LEGO item, whereas position~2 opposes it. 
Finally, we edit existing history activations while keeping model parameters and candidate PIDs fixed to probe \textbf{feature-dependent retrieval responses}. 
For $f_{4260}$, a $+0.50$ shift in standardized log-activation increases exposure@10 by 1.365 percentage points relative to matched control edits on 332 eligible test users (Appendix~C.2). 
\begin{findingsbox}{Findings for RQ4}
\textbf{Finding 1: Vocabulary sharing without refitting.}
Frozen pretrained dictionary supports competitive recommendation across benchmarks. 

\par\smallskip
\textbf{Finding 2: Adaptation within a shared vocabulary.}
PID reduces the collisions while preserving the pretrained feature-index identities.
\end{findingsbox}

\begin{findingsbox}{Findings for RQ5}
\noindent\textbf{Finding 1: Inspectable ranking evidence.}
Score decomposition links readable features to evidence supporting or opposing a candidate.

\noindent\textbf{Finding 2: Local steering effects.}
Activation edits can change exposure, with responses varying across features and doses.
\end{findingsbox}


%
%
%
%
%
%
\section{Conclusion}
We propose API to address catalog-dependent vocabulary learning in generative recommendation by adapting a pretrained interpretable dictionary into a shared vocabulary. 
Through PID construction and SPA, API connects compact item representations and user preference learning with interpretable retrieval evidence, opening a practical path for applying mechanistic interpretability to recommendation. 
Learning more general-purpose interpretable vocabularies and fully realizing their potential for recommendation remain promising directions for future research.

\textbf{LLM Usage.}
Assist with baseline reproduction, refining grammar, and reviewing mathematical proofs. The authors take full responsibility for the final manuscript and experimental results.

\bibliographystyle{ACM-Reference-Format}
\bibliography{WWW2027}

\appendix
\section{Experimental Details}
\label{app:experimental_details}

\subsection{Data and Recommendation Evaluation}
\label{app:exp_data}\label{app:evaluation}
We retain the supplied Amazon splits and item orderings, with one validation and one test target per user.
Validation uses up to 20 preceding training interactions; test uses up to 20 preceding interactions, including the validation item.
The TIGER sampler converts 32 user sequences into at most 512 contiguous next-item windows per micro-batch, each containing up to 29 history items and one target.
Evaluation uses unaugmented queries in batches of eight, including the final partial batch.

\begin{table}[!ht]
\centering\small
\setlength{\tabcolsep}{4pt}
\caption{Dataset statistics after pre-processing.}
\label{tab:exp_datasets}
\begin{tabular}{@{}lrrrr@{}}
\toprule
Dataset & Users & Items & Interactions & Train items\\
\midrule
Beauty & 22,363 & 12,101 & 198,502 & 12,068\\
Sports & 35,598 & 18,357 & 296,337 & 18,284\\
Toys & 19,412 & 11,924 & 167,597 & 11,876\\
\bottomrule
\end{tabular}
\end{table}

For target rank $r_u$, Recall@$k$ and NDCG@$k$ average $\mathbf1[r_u\le k]$ and $\mathbf1[r_u\le k]/\log_2(r_u+1)$ over users, respectively, for $k\in\{5,10\}$; unretrieved targets contribute zero.
Evaluation uses item identities, without sampled negatives or seen-item filtering.
PR scores the full catalog exactly; AR uses beam search over the catalog trie with width 50.
Checkpoints selected by full-validation Recall@5 are reloaded for full-test evaluation.

\subsection{Pretrained Resources and PID Construction}
\label{app:exp_sids}
We extract GPT-2-small~\cite{radford2019language} with Bloom's SAE through SAELens (\texttt{gpt2-small-res-jb}, hook \texttt{blocks.8.\allowbreak hook\_resid\_pre}).
Item text concatenates title, brand, category, and price, truncated to 128 tokens.
We pool token-level SAE activations excluding special and padding tokens, with decoder-norm scaling as in Eq.~\eqref{eq:pid-activation}.
Recommendation backbones are randomly initialized.

Each item contributes its top-256 positive features.
We estimate $\rho_j=(n_j+1)/(N_{\rm tr}+1)$ from training items, where $n_j$ counts candidate sets containing $j$ and $N_{\rm tr}$ counts training items; compact remapping preserves original feature identities.
SPA orders selected features by decreasing activation, breaking ties by index; direct PID-PR uses training-frequency ordering.
Ordering is fixed within SPA ablations.
A nonsemantic suffix disambiguates shared supports, yielding $\tilde c_i=(c_i,\nu_i)$; strength profiles and collision diagnostics exclude it.

History encoding standardizes $\log(1+a_{ij})$ using featurewise training means $\mu_j$ and standard deviations $\sigma_j$, smoothed with 16 global pseudo-observations and a standard-deviation floor of 0.05.
Values clipped to $[-2,2]$ interpolate eight learned gain/shift embeddings that modulate feature embeddings before GELU.
Four priority-selected anchors incorporate projected mean top-32 context and position embeddings; a fifth encodes the suffix.
Layer normalization yields five history tokens per item for every output $K$.
Output strengths normalize selected log activations without history calibration.

\subsection{Backbones, Training, and Baselines}
\label{app:implementation}\label{app:exp_optimization}\label{app:baselines}
\begin{table}[!ht]
\centering\small
\setlength{\tabcolsep}{3pt}
\caption{Recommendation backbones. E/D: encoder/decoder layers. GPT-2 AR and PR are separate capacity cohorts.}
\label{tab:exp_backbones}
\begin{tabular}{@{}lccc@{}}
\toprule
Setting & T5 AR/PR & GPT-2 PR & GPT-2 AR\\
\midrule
Layers (E/D) & 4/4 & 0/4 & 0/4\\
Width / heads & 128/6 & 448/8 & 128/8\\
K/V per head & 64 & 56 & 16\\
FFN widths & 1,024/1,024 & 1,792 & 1,024\\
Activation & ReLU & SiLU & GELU\\
Dropout & 0.15 & 0.15 & 0.15\\
\bottomrule
\end{tabular}
\end{table}
T5 uses relative positions and a three-linear-layer FFN; GPT-2 uses learned positions (512 for PR, history-plus-target capacity for AR). RPG-style PR baselines flatten SIDs, use position-specific vocabularies and tied input/output embeddings, and apply normalized-dot-product prediction heads at temperature 0.07. T5-SPA-PR uses a noncausal query decoder; GPT-2-SPA-PR appends learned role queries under causal attention. Both predict all roles in one pass without target-token inputs. AR uses shifted teacher forcing and candidate-specific prefixes. The history-only strength prediction is shared across candidates.

The standard recommendation recipe uses seed 45, the gradient accumulation over 16 microbatches, without learning-rate warmup or scheduling. 
Full validation occurs every 1,600 microbatches; early stopping allows ten unimproved checks, with limits of ten epochs and 320,000 optimizer steps. These limits concern recommendation training, not tokenizer fitting. Equal data protocols do not imply equal parameter counts or early-stopped exposure.

Full PR-SPA uses $(\lambda,\beta)=(1,8)$ in Eq.~\eqref{eq:spa-objective}. The four supervision settings are $(0,0)$, $(0,8)$, $(1,0)$, and $(1,8)$; all retain activation-aware histories and SID prediction. 
PR preference loss is positive-item cross-entropy after a temperature-one softmax over all catalog scores. 
The AR implementation instead uses the positive and four distinct uniform negatives, scoring each under its own prefix; this five-item objective is not exact full-catalog cross-entropy. 
T5-SPA-AR settings use $(\lambda,\beta)=(0.25,2)$ on Beauty and $(1,8)$ on Sports/Toys. At inference, AR normalizes over the full token vocabulary before masking illegal trie edges and applies strengths at complete leaves.

OPQ rotates and product-quantizes mean-pooled PLM residuals; the matched PR recipe uses $K$ subquantizers, 256 centroids each, 25 rotation/PQ iterations, and no additional PCA or whitening. Residual-tokenizer baselines preserve stage-specific code identities: RVQ learns residual codebooks, RKMeans applies residual clustering, and RQ-VAE learns an encoder/decoder around residual quantization. The verified GPT-2 residual-AR assets use four semantic indices plus one suffix. Their original-recipe tokenizers fit all catalog item contents (without held-out interaction labels), with seed 42 and batch size 2,048. RVQ uses per-stage residual normalization, Adagrad at $10^{-3}$, commitment weight 0.25, and 3,000 updates; RKMeans uses the same residual normalization, SGD at 0.5, and 30 updates. RQ-VAE uses batch/L2-normalized inputs, a $768$--$768$--$256$--$128$--$64$ encoder and mirrored decoder, no per-stage normalization, and the same Adagrad rate, commitment weight, and update budget as RVQ. Random recommendation codes are distinct from the random-positive-feature selection control used below. The source-transfer recipes in Appendix~\ref{app:reuse} form a separate diagnostic cohort, not additional replications of the main-table tokenizer fits.

\paragraph{Baselines.}
\textit{(1) Item tokenization.} 
Random assigns codes independently of item content.
OPQ~\cite{hou2025generating} jointly optimizes a rotation and subspace codebooks. 
RVQ learns successive residual codebooks, whereas R-KMeans clusters successive residuals; we follow their GRID implementations~\cite{ju2025grid}. 
RQ-VAE~\cite{lee2022autoregressive} learns an encoder and decoder around residual quantization, and HiD-VAE~\cite{fang2025hidvae} constructs hierarchical, disentangled semantic identifiers. 
Native P-SAE selects the strongest pretrained features, whereas R-SAE learns a sparse dictionary from item embeddings. 
The random-positive control samples from the active feature pool and differs from Random item codes. 
\textit{(2) Recommendation frameworks.} 
Our tokenizer comparisons use TIGER-style AR~\cite{rajput2023recommender}
and RPG-style PR~\cite{hou2025generating} with T5/GPT-2 backbones:
the former predicts SID tokens sequentially, while the latter predicts
the code positions in parallel.
The complete-framework comparison additionally includes
LC-Rec~\cite{zheng2024adapting}, which aligns collaborative semantics
with language-model representations;
UTGRec~\cite{zheng2026universal}, which learns a shared tokenizer
for transfer across recommendation domains;
Ghost~\cite{yin2026echoes}, which addresses popularity bias through
asymmetric unlikelihood optimization and skeleton-founded tokenization. 
DiGR~\cite{zhao2026simgr} instead uses SIDs as item representations
for direct item-level scoring.
GRU4Rec~\cite{liu2026setllm} and BERT4Rec~\cite{tan2024idgenrec}
provide conventional sequential-recommendation references using recurrent
and bidirectional masked-item modeling.

\section{Proofs and Score Derivatives}
\label{app:spa-proofs}

\subsection{Proof of Proposition~\ref{prop:pid-separation}}
\label{app:pid-separation}

\begin{proof}
Fix a group $\mathcal G$ of items sharing the token sequence
$c=(c_1,\ldots,c_K)$.
Let
$\bar{\mathbf h}_i
=|\mathcal T_i|^{-1}\sum_{t\in\mathcal T_i}\mathbf h_i^{(t)}$
denote the mean PLM representation over the positions used in
Eq.~\eqref{eq:pid-activation}.
We assume nonnegative activations with positive sums
$\omega_i=\sum_{r=1}^{K}a_{i,c_r}$, ensuring that
$\mathbf q_i$ is well-defined.

Write $\mathbf u_j=\mathbf d_j/\|\mathbf d_j\|_2$ for nonzero
decoder columns, and set $\mathbf u_j=\mathbf 0$ otherwise.
By Eq.~\eqref{eq:pid-activation}, the mean SAE reconstruction
using only the selected features is
\begin{equation}
    \tilde{\mathbf h}_i
    =
    \mathbf b_{\mathrm{dec}}
    +\sum_{r=1}^{K}\mathbf u_{c_r}a_{i,c_r}.
\end{equation}
The assumptions in the proposition give
$\|\bar{\mathbf h}_i-\tilde{\mathbf h}_i\|_2\leq\epsilon$
for every $i\in\mathcal G$, while the sums $\omega_i$ lie
in an interval of width at most $\delta$.

Consider $n_{\max}$ items sharing the same PID.
The case $n_{\max}=1$ is immediate.
For any two distinct items $i,i'$ in this set,
$\mathbf q_i=\mathbf q_{i'}$.
Equation~\eqref{eq:pid-profile} therefore implies that
\[
    \log(1+a_{i,c_r})
    =
    \alpha\log(1+a_{i',c_r})
    \qquad\text{for all }r
\]
for some $\alpha>0$.
Since the log strengths are nonnegative and the logarithm
is strictly increasing, all differences
$a_{i,c_r}-a_{i',c_r}$ have the same sign.
Using $\|\mathbf u_{c_r}\|_2\leq1$ and the triangle inequality,
we obtain
\begin{equation}
    \begin{aligned}
        \Delta
        &\leq
        \|\bar{\mathbf h}_i-\bar{\mathbf h}_{i'}\|_2\\
        &\leq
        2\epsilon+
        \|\tilde{\mathbf h}_i-\tilde{\mathbf h}_{i'}\|_2\\
        &\leq
        2\epsilon+
        \sum_{r=1}^{K}|a_{i,c_r}-a_{i',c_r}|\\
        &=
        2\epsilon+|\omega_i-\omega_{i'}|.
    \end{aligned}
\end{equation}
Thus, any two such items have activation sums separated
by at least $\Delta-2\epsilon>0$.

Order these sums as
$\omega_{(1)}\leq\cdots\leq\omega_{(n_{\max})}$.
Summing the lower bounds on adjacent gaps gives
\begin{equation}
    (n_{\max}-1)(\Delta-2\epsilon)
    \leq
    \omega_{(n_{\max})}-\omega_{(1)}
    \leq
    \delta.
\end{equation}
Rearranging and using the integrality of $n_{\max}$ yields
\[
    n_{\max}
    \leq
    1+\left\lfloor
    \frac{\delta}{\Delta-2\epsilon}
    \right\rfloor.
\]
In particular, $\delta<\Delta-2\epsilon$ forces
$n_{\max}=1$, so all items in $\mathcal G$ have distinct PIDs.
\end{proof}

\subsection{Proof of Proposition~\ref{prop:spa-strength-margin}}
\label{app:spa-strength-margin}

\begin{proof}
Fix a user $u$, a target item $i^+$, and a competing item $i^-$.
Let $\hat{\mathbf q}_u$ denote SPA's prediction of the next item's relative strengths given $\mathbf H_u$, and let $\mathbf q_{i^+}$ and $\mathbf q_{i^-}$ be the strength vectors recorded in their PIDs.

Write $\mathbf x=\sqrt{\hat{\mathbf q}_u}$, $\mathbf y=\sqrt{\mathbf q_{i^+}}$, and $\mathbf z=\sqrt{\mathbf q_{i^-}}$, with square roots taken elementwise.
Since the strength vectors are nonnegative and sum to one, $\|\mathbf x\|_2=\|\mathbf y\|_2=\|\mathbf z\|_2=1$.
The definitions of $\Delta_q$ and $\varepsilon_u$ therefore give
\begin{equation}
\Delta_q=1-\mathbf y^\top\mathbf z=\frac{1}{2}\|\mathbf y-\mathbf z\|_2^2,\qquad \varepsilon_u=\frac{1}{2}\|\mathbf x-\mathbf y\|_2^2.
\end{equation}

By Eq.~\eqref{eq:spa-strength-distance}, the difference between the two strength penalties is
\begin{equation}
D_{\mathrm H}(\hat{\mathbf q}_u,\mathbf q_{i^-})-D_{\mathrm H}(\hat{\mathbf q}_u,\mathbf q_{i^+})=\mathbf x^\top(\mathbf y-\mathbf z).
\end{equation}
Writing $\mathbf x=\mathbf y+(\mathbf x-\mathbf y)$ and using $\mathbf y^\top(\mathbf y-\mathbf z)=\Delta_q$, the Cauchy--Schwarz inequality yields
\begin{equation}
\mathbf x^\top(\mathbf y-\mathbf z)\geq\Delta_q-\|\mathbf x-\mathbf y\|_2\|\mathbf y-\mathbf z\|_2=\Delta_q-2\sqrt{\Delta_q\varepsilon_u}.
\end{equation}

Substituting this bound into Eq.~\eqref{eq:spa-item-score} gives
\begin{equation}
s_\theta(u,i^+)-s_\theta(u,i^-)\geq\ell_u(i^+)-\ell_u(i^-)+\beta\big(\Delta_q-2\sqrt{\Delta_q\varepsilon_u}\big).
\label{eq:spa-general-strength-margin}
\end{equation}
This inequality holds for general candidate pairs.
Under the proposition's assumption $c_{i^+}=c_{i^-}$, their sequence likelihoods are equal, so the bound reduces to Eq.~\eqref{eq:spa-strength-margin}.

Finally, $\varepsilon_u<\Delta_q/4$ implies $\Delta_q-2\sqrt{\Delta_q\varepsilon_u}>0$.
Together with $\beta>0$, this gives $s_\theta(u,i^+)>s_\theta(u,i^-)$, establishing the claimed ranking.
\end{proof}

\section{Diagnostic and Interpretation Protocols}
\subsection{Vocabulary Quality}
\label{app:representation-details}\label{app:reuse}
For a catalog of $N$ items and a suffix-free key $g(i)$, let $n_i$ count items sharing that key.
The collision rate is $N^{-1}\sum_i\mathbf{1}[n_i>1]$.
SAE keys are unordered positive-feature supports; quantizer keys retain stage identity.
Strengths are excluded, so this statistic measures support ambiguity rather than equality of complete PIDs.
Native, priority-based, and random-positive selection share a top-32 candidate pool; random selection samples without replacement.

Content alignment uses 2,000 fixed queries per domain (seed 45), the full catalog excluding each query, and a shared random tie order (seed 46).
For representation and reference top-20 neighborhoods $\mathcal N_i$ and $\mathcal N_i^{\rm ref}$,
\begin{equation}
\operatorname{Overlap}@20=\frac{1}{20|Q|}\sum_{i\in Q}|\mathcal N_i\cap\mathcal N_i^{\rm ref}|.
\end{equation}
References use cosine similarity of title TF--IDF (lowercase, English stop-word removal, unigrams/bigrams, sublinear TF, smoothed IDF) or mean-pooled PLM residuals.
Index-only neighborhoods use support Jaccard; activation-aware neighborhoods use cosine similarity of sparse, decoder-rescaled raw activations, distinct from SPA's normalized log strengths.
Permutation moves complete feature--activation profiles across items with content fixed, preserving collision multiplicities.

Reuse diagnostics freeze source-fitted tokenizers across targets, using four semantic tokens and the same 10,000-item subset per target for collision evaluation.
R-SAE uses 4,096 features and 3,000 updates; RQ-VAE uses 3,000 updates.
RVQ/R-KMeans use globally RMS-scaled residuals without per-stage normalization, taking the first four stages of 32-stage encoders trained for 100 updates per stage.
OPQ fits source rotation/codebooks.
PID retains the frozen P-SAE dictionary and estimates IDF from target training-item contents.
Comparisons match token counts, not bit rates; downstream recommenders are trained on target interactions.

\subsection{Semantic Annotations and Steering}\label{app:steering}
Neuronpedia supplies external, automatically generated feature descriptions, unused as model inputs.
Title agreement is the lexical-match rate among feature-containing training items; lift divides it by the training-catalog match rate.
For $f_{4260}$, the case-insensitive whole-word rule matches hunt, hunting, rifle, rifles, gun, guns, pistol, shotgun, holster, ammo, ammunition, or scope.

Steering freezes T5-PR Full SPA ($K=8$) and edits six prespecified features per domain at doses $d\in\{0,\pm0.25,\pm0.50\}$.
Existing positive history activations are shifted by $d$ in unclipped standardized log coordinates, bounded below by the positive-activation floor, clipped to $[-2,2]$, and re-encoded; supports, anchors, suffixes, and candidate PIDs stay fixed.
Eligible users have the target feature in at least two history items and an anchor responsive in both directions.
Within this cohort, three deterministic control features match history frequency, anchors, and activation statistics; per-query scaling matches post-clipping perturbation norms.

Exposure@10 measures the fraction of retrieved items whose PID contains the target feature.
Changes from unedited histories are reported in percentage points; Toys equally averages six feature-specific effects.
Doses are validation-selected to retain at least 98\% Recall@10 on both the full split and matched cohort.
Curve bands show min--max over three runs.
Paired target-minus-control effects average controls within users; nominal 95\% percentile intervals use 500 paired user-bootstrap resamples without correction across 36 feature--direction comparisons.
For Sports $f_{4260}$ at $d=0.50$, the effect is $+1.365$ points on 332 test users, with interval $[0.652,2.189]$.

\section{Additional Results}\label{app:res}
\suppressfloats[t]

\textbf{Vocabulary reuse.}
Figure~\ref{fig:app-sports-reuse} extends the source-domain comparison to Sports.
Across the three target catalogs, PID has substantially fewer collisions than native P-SAE while using the same frozen dictionary.
RQ-VAE attains lower collision rates in this comparison, while PID supports item distinction without extra vocabulary fitting.

\begin{figure}[!htbp]
    \centering
    \includegraphics[width=\linewidth]{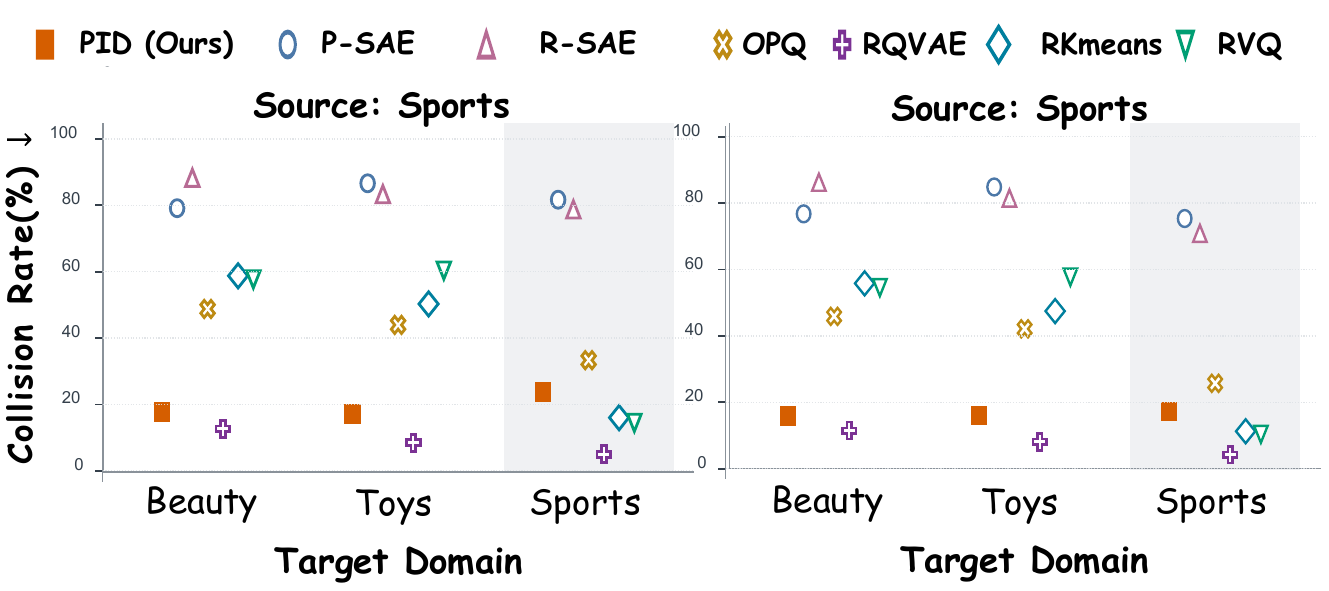}
    \caption{Vocabulary reuse from Sports. Collision-affected items (\%) are reported; shaded columns mark the source domain used to fit learned tokenizers. PID reuses the frozen P-SAE dictionary with target-catalog IDF.}
    \label{fig:app-sports-reuse}
\end{figure}

\textbf{Collisions under compact budgets.}
Figure~\ref{fig:app-budget-collisions} compares feature selection and quantization at $K\in\{4,8,16\}$.
At $K=4$, PID reduces the collision-affected rate from roughly 80\% under native selection to 14.64\%--21.48\% across datasets; the gap narrows as the budget increases.
Random-positive selection and RQ-VAE often yield still fewer collisions.
Together with the content-alignment results in the main text, this comparison highlights why low collision rates alone are insufficient to assess semantic item representations.

\begin{figure}[!htbp]
    \centering
    \includegraphics[width=\linewidth]{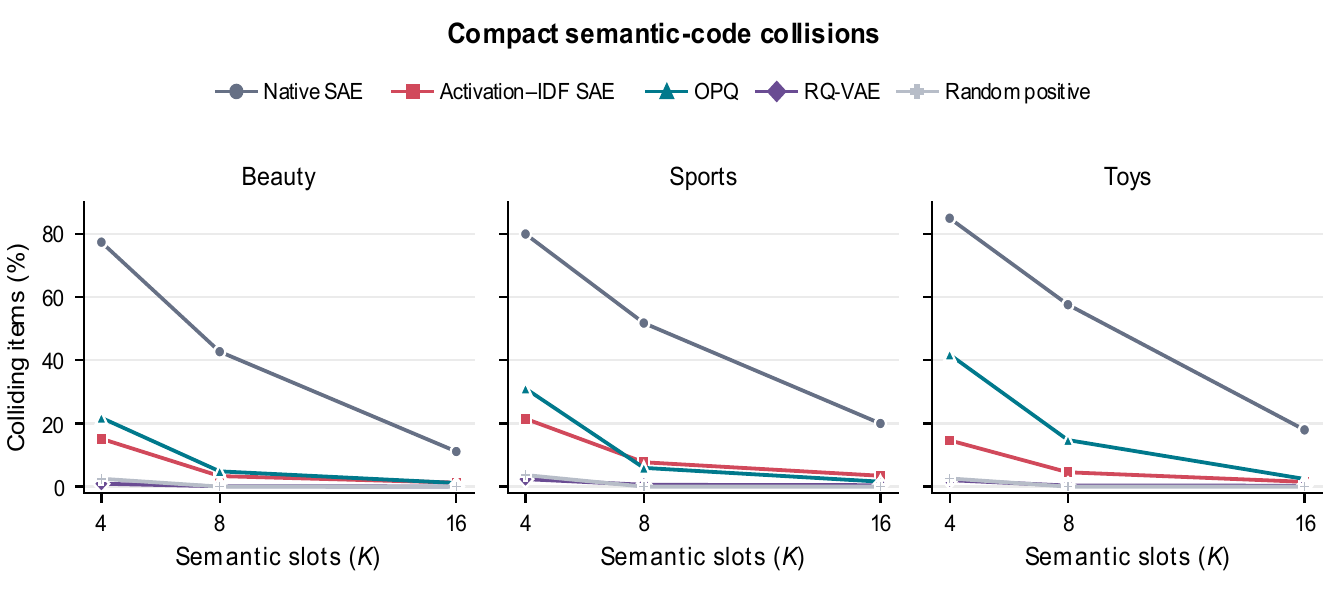}
    \caption{Collision-affected items (\%) across semantic-token budgets on Beauty, Sports, and Toys. Activation--IDF SAE denotes PID; random-positive selection samples active SAE features. Strengths and disambiguation suffixes are excluded from collision statistics.}
    \label{fig:app-budget-collisions}
\end{figure}

\textbf{Additional semantic examples.}
The Beauty feature $f_{14730}$ and Toys feature $f_{13542}$ further connect pretrained descriptions with item content.
At $K=4$, title agreement reaches 92.4\% for brushes and 99.6\% for LEGO, with lifts of $18.31\times$ and $13.77\times$, respectively.
Their top-activating titles exhibit the same patterns, providing complementary evidence for interpreting the selected features.

\end{document}